\documentclass{article}
\usepackage{setspace}
\usepackage{braket}
\usepackage{bbm}
\usepackage{tcolorbox}
\usepackage{graphicx}
\usepackage{amsthm}
\usepackage{hyperref}
\hypersetup{
    colorlinks,
    citecolor= orange,
    filecolor= blue,
    linkcolor= blue,
    urlcolor= green
}
\usepackage{multicol}
\usepackage{mathrsfs}
\usepackage{appendix}
\usepackage{amssymb}
\usepackage{relsize}
\usepackage{amsthm}
\usepackage{ amssymb }
\usepackage{amsmath}
\usepackage[utf8]{inputenc}
\usepackage[english]{babel}
\usepackage[a4paper, total={6in, 9.2in}]{geometry}
\usepackage{amsmath}

\newtheorem{definition}{Definition}
\newtheorem{definition2}{Lemma}

\newtheorem{definition4}{Theorem}
\newtheorem{definition5}{Remark}

\newtheorem{Co}{Corollary}
\usepackage[T1]{fontenc}
\usepackage{lmodern}

\usepackage{tikz-cd}

\begin{titlepage}

\title{SBSCV}
\author{Alberto Acevedo}

\date{2023}

\end{titlepage}

\begin{document}

\thispagestyle{plain}
\begin{center}
   
    \textbf{Spectrum Broadcast Structures from von Neumann type interaction Hamiltonians}

    \vspace{0.4cm}
    
    \textbf{Alberto Acevedo$^{1}$, Jaros\l aw K. Korbicz$^{2}$, Janek Wehr $^{1}$ }

      \vspace{0.4cm}
      \textbf{$^{1}$The University of Arizona, Tucson, USA}
      
      \vspace{0.2cm}
      
      \textbf{$^{2}$Center for Theoretical Physics, Polish Academy of Sciences, Warsaw, Poland}

    \vspace{0.9cm}
    \textbf{Abstract:}
  In this paper, we contribute to the mathematical foundations of the recently established theory of Spectrum Broadcast Structures (SBS). These are multipartite quantum states, encoding an operational notion of objectivity and exhibiting a more advanced form of decoherence. We study SBS in the case of a central system interacting with $N$ environments via the von Neumann-type measurement interactions, ubiquitous in the theory of open quantum systems. We state and prove a novel sufficient condition for SBS to arise dynamically 
for finite-dimensional systems. The condition is based on the Gram-Schmidt orthogonalization rather than on the Knill-Barnum error estimation used before.
\end{center}
\onehalfspacing

\

\section{Introduction}

\;\;\;In recent times significant attention has been given to decoherence processes, not only because of the rapid development of quantum technologies but also in the connection 
with such fundamental problems as the quantum-to-classical transition  (we refer to e.g. \cite{schloss} for a comprehensive introduction). In this context, especially interesting is the idea of so-called quantum Darwnism \cite{Zurek}\cite{Zurek2}. It recognizes that one of the features of the macroscopic world is its objective character and seeks to explain it from quantum theory.  A certain property, e.g. a position of a system of interest, becomes objective if it leaks into the environment during the decoherence process and is redundantly stored there in many copies, making it accessible to multiple observations. 
This idea fills the foundational gap between the subjective, observer-dependent character of quantum observations and the objective character of the macroscopic world. 
Later a stronger form of quantum Darwinism was proposed, named \emph{Spectrum Broadcast Structures} (SBS) \cite{JKone} \cite{JKtwo} \cite{JKthree} \cite{kor5}, where objectivity is encoded directly in the structure of multipartite quantum states.
SBS states satisfy the information-theoretic conditions used to test for quantum Darwinism but the opposite is not always true
\cite{Olaya}. In this work, we study the formation of SBS states for an important class of systems, interacting with the environment via von Neumann type measurement interaction.

We begin by quoting an ideal definition of objectivity, proposed in \cite{Zurek}:

\begin{definition}
\label{def:objectivity} A state of the system $S$ exists objectively if many observers can find out the state of $S$ independently, and without perturbing it.
\end{definition}

Spectrum Broadcast Structures, defined below, are a mathematical formulation of Definition \ref{def:objectivity}. 

\begin{definition}{SBS:}
\label{eqn:turk4}
A Spectrum Broadcast Structure is a multipartite state of a central system $S$ and an environment 
 $E$, consisting of sub-environments $E^{1}, E^{2},..., E^{N_{E}}$:
\begin{equation}
\label{eqn:morocco}
\boldsymbol{\hat{\rho}} = \sum_{i}p_{i}|i\rangle\langle i|\otimes \bigotimes_{k= 1}^{N_{E}}\boldsymbol{\hat{\rho}}_{i}^{E^{k}}
\end{equation}
where $\{|i\rangle\}_{i}$ is some basis in the system's space, called the pointer basis \cite{Zurek3}, $p_{i}$ are probabilities summing up to one, and all states $\boldsymbol{\hat{\rho}}_{i}^{E^{k}}$ are perfectly distinguishable in the following sense:
\begin{equation}\label{eq:dist}
F\big(\boldsymbol{\hat{\rho}}_{i}^{E^{k}},\boldsymbol{\hat{\rho}}_{j}^{E^{k}}\big) =  0
\end{equation}
for all $i\neq j$ and for all $k = 1,...,N_{E}$, where $F(\cdot,\cdot)$ is the quantum fidelity defined as $F\big(\boldsymbol{\hat{\rho}},\boldsymbol{\hat{\sigma}}\big):=\big\|\sqrt{\boldsymbol{\hat{\rho}}}\sqrt{\boldsymbol{\hat{\sigma}}}\big\|_{1}^{2}$ \cite{fuchs}.
\end{definition}

\;\;\;In this paper, we study the convergence of multipartite states evolving dynamically in time as described in Section \ref{sec:monitoring}, to an SBS state as defined in Definition \ref{eqn:turk4}. For any fixed time, the multipartite states in question will not be SBS states. However, trace distance estimates developed in this work may be employed for estimating the proximity of a multipartite state to an SBS state. 

The SBS states comply with Definition \ref{def:objectivity} in the following way:  measuring each of the environments $E^{k}$ independently via projections onto the supports of the states  $\boldsymbol{\hat{\rho}}_{i}^{E^{k}}$, one obtains an outcome $i$ with certainty, due to the orthogonality of the supports for different $i$, guaranteed by the perfect distinguishability condition \eqref{eq:dist}. These outcomes,
identifying in which of the pointer states $|i\rangle$ the system is, are the same as if one  measured the system directly and moreover the same for every sub-environment $E^{k}$, which is due to the particular structure \eqref{eqn:morocco}. This leads to an agreement among observators, which is a cornerstone of the notion of objectivity. Moreover, as one can easily check, the average post-measurement state, averaged over the outcomes, is again the same state \eqref{eqn:morocco}. This can be viewed as a formalization of the non-disturbance condition and is actually related to the famous  Bohr vs. Einstein-Podolsky-Rosen  
debate over the completeness of quantum theory \cite{kor5}. In \cite{JKone} the opposite was also proven, namely that starting from Definition \ref{def:objectivity} and using an additional assumption of so-called strong independence, one obtains the SBS form as the only compatible state (for relaxations without the strong independence assumption, we refer to \cite{kor5}).  

An interesting problem arises. Namely, how to recognize if during the evolution the state of a multipartite system approaches the SBS form  \eqref{eqn:morocco}, implying the state of the system becomes objective in the sense of Definition \ref{def:objectivity}. This problem was studied for a number of open system models (see e.g. \cite{JKtwo,models1,models2}), as well as in a more general setting in \cite{JKthree}. There, an upper bound was found on the trace distance between the actual state of the system and the nearest SBS state under the assumption of von Neumann-type measuring interaction, which we define in Section \ref{sec:monitoring} and which is by far the most common form of an interaction used in the theory of open quantum systems \cite{schloss} (sometimes called pure dephasing interaction). However, the derivation in \cite{JKthree} uses the so-called Knill-Barnum bound \cite{knill} on the probability of discrimination error outside its known range of applicability. Here, we will take another path, circumventing the use of the Knill-Barnum theorem.

\section{Statement of the problem and existing results}
\;\;\;Before proceeding we introduce the concepts of \emph{positive operator-valued measure} (POVM) and \emph{quantum measurement} and mention that all of the Hilbert spaces in this paper are assumed to be separable. 
\begin{definition}{POVM:}
Consider an arbitrary Hilbert space $\mathscr{H}$. A trace-preserving POVM is a set of positive semi-definite operators $\big\{\mathbf{\hat{M}}_{i}^{\dagger}\mathbf{\hat{M}}_{i}\big\}_{i}$ acting in $\mathscr{H}$ that sum to the identity operator. i.e. 
\begin{equation}
\label{eqn:res}
\sum_{i}\mathbf{\hat{M}}_{i}^{\dagger}\mathbf{\hat{M}}_{i} = \mathbb{I}_{\mathscr{H}}
\end{equation}
\end{definition}
\begin{definition5}
The POVM may consist of an uncountable set of semi-definite operators as well. In such case the analogous set of operators, e.g. $\mathbf{\hat{M}}_{x}^{\dagger}\mathbf {\hat{M}}_{x}$ ($x\in\mathbb{R}$) must satisfy the constraint analogous to (\ref{eqn:res}), i.e.
\begin{equation}
\int\mathbf{\hat{M}}_{x}^{\dagger}\mathbf{\hat{M}}_{x}dx = \mathbb{I}_{\mathscr{H}}
\end{equation}
\end{definition5}
\begin{definition}{Quantum Measurement:}
Fix a Hilbert space $\mathscr{H}$ and let $\mathcal{S}\big(\mathscr{H}\big)$ be the set of density operators acting in $\mathscr{H}$. Let $\boldsymbol{\hat{\rho}}\in \mathcal{S}\big(\mathscr{H}\big)$, the state of a system after a measurement is 
\begin{equation}
\sum_{i}\mathbf{\hat{M}}_{i}\boldsymbol{\hat{\rho}}\mathbf{\hat{M}}_{i}^{\dagger}
\end{equation}
where the set $\big\{\mathbf{\hat{M}}_{i}^{\dagger}\mathbf{\hat{M}}_{i}^{\dagger}\big\}_{i}$ is a POVM.

\end{definition}

In \cite{JKone} it is argued that SBS satisfies the desired definition of objectivity and that it is the only structure that satisfies it. The argument for why the observer monitoring $E^{l}$ may find out the state of $S$ is the following. Let us analyze the local state pertaining to $E^{l}$; to do this we partially trace out the degrees of freedom pertaining to the system $S$ and all of the environments $E^{k}$ with the exception of the $l$th environment. i.e. from (\ref{eqn:morocco}) we obtain  
\begin{equation}
\label{eqn:themixtureparty}
\sum_{i}p_{i}\Bigg(\prod_{k\neq l}Tr_{E^{k}}\big\{\boldsymbol{\hat{\rho}}_{i}^{E^{k}}  \big\}\Bigg)\langle i|i\rangle\boldsymbol{\hat{\rho}}_{i}^{E^{l}} = \sum_{i}p_{i}\boldsymbol{\hat{\rho}}_{i}^{E^{l}}.
\end{equation}
Notice that this is a mixed state. If $F\big(\boldsymbol{\hat{\rho}}_{i}^{E^{l}}, \boldsymbol{\hat{\rho}}_{j}^{E^{l}}   \big) = 0$ for all $i\neq j$ then the associated quantum state discrimination problem may be fully solved. This means that there exists a POVM which the observer monitoring the environment $E^{l}$ may utilize to conduct measurements on $E^{l}$ yielding perfect distinguishability between the possible outcomes of the mixture (\ref{eqn:themixtureparty}). Furthermore, the state $\boldsymbol{\hat{\rho}}_{i}^{E^{l}}$ is correlated with the state  $\big| i\big \rangle \big\langle i\big|$ of $S$ in the sense that when $S$ is found to be in the state $\big|i\big\rangle \big\langle i \big|$ the $l$th environment will be found in the state $\boldsymbol{\hat{\rho}}_{i}^{E^{l}}$. Owing to the perfect distinguishability between the states $\boldsymbol{\hat{\rho}}_{i}^{E^{l}}$ for all $i$, there is no ambiguity regarding the state of $S$ given that $E^{l}$ is found to be in the state $\boldsymbol{\hat{\rho}}_{i}^{E^{l}}$.\\ 

To argue non-disturbance, we first re-emphasize that the "can find out" in Definition \ref{def:objectivity} formally means that for every $E^{k}$ there exists a POVM $\{\mathbf{\hat{E}}^{E^{k}}_{i}\}_{i}$ that solves the respective local QSD problem, i.e. that discriminates perfectly the mixture (\ref{eqn:themixtureparty}).   
$\{\bigotimes_{k=1}^{N_{E}}\mathbf{\hat{E}}^{E^{k}}_{i_{k}}\}_{i_{1},i_{2},...,i_{N_{E}}}$ will henceforth denote a POVM acting on $\mathcal{S}\big(\mathscr{H}_{S}\bigotimes_{k=1}^{N_{E}}\mathscr{H}_{E^{k}}\big)$. Suppose the POVM optimally solving the local QSD problem for each environment $E^{l}$ does it in a non-perturbing way, i.e. the state after the associated measurement doesn't change. Then the measurement associated with the POVM may be shown to also be non-disturbing $\{\bigotimes_{k=1}^{N_{E}}\mathbf{\hat{E}}^{E^{k}}_{i_{k}}\}_{i_{1},i_{2},...,i_{N_{E}}}$, i.e. 
\begin{equation}
\label{eqn:Canada2}
\sum_{i}p_{i}|i\rangle\langle i|\otimes \bigotimes_{k= 1}^{N_{E}}\boldsymbol{\hat{\rho}}_{i}^{E^{k}} =  \sum_{i_{1}}\sum_{i_{2}}...\sum_{i_{N_{E}}}\Bigg( \sum_{i}p_{i}|i\rangle\langle i|\otimes \bigotimes_{k= 1}^{N_{E}}\mathbf{\hat{M}}^{E^{k}}_{i_{k}}\boldsymbol{\hat{\rho}}_{i}^{E^{k}}\Big(\mathbf{\hat{M}}^{E^{k}}_{i_{k}}\Big)^{\dagger}\Bigg) 
\end{equation}
where $\Big(\mathbf{\hat{M}}^{E^{k}}_{i_{k}}\Big)^{\dagger}\mathbf{\hat{M}}^{E^{k}}_{i_{k}}=\mathbf{\hat{E}}^{E^{k}}_{i_{k}}$ is the POVM perfectly discriminating the mixture $\sum_{i}p_{i}\boldsymbol{\hat{\rho}}_{i}^{E^{k}}$. To show (\ref{eqn:Canada2}) note that perfect distinguishability  of the $\boldsymbol{\hat{\rho}}_{i}^{E^{k}}$ for all $k$ implies that we may devise a POVM $\{\mathbf{\hat{E}}^{E^{k}}_{i_{k}}\}_{i_{k}}$ such that  
\begin{equation}
\label{eqn:Canada1}
\mathbf{\hat{M}}^{E^{k}}_{i_{k}}\boldsymbol{\hat{\rho}}_{i}^{E^{k}}\Big(\mathbf{\hat{M}}^{E^{k}}_{i_{k}}\Big)^{\dagger} = \delta_{i_{k}i}.
\end{equation}
owing to the non-overlapping of the supports of the $\boldsymbol{\hat{\rho}}_{i}^{E^{k}}$.
With (\ref{eqn:Canada1}) in mind, we may estimate the trace distance in (\ref{eqn:Canada2}). 
\begin{equation}
\label{eqn:Canada5}
\frac{1}{2}\bigg\| \sum_{i}p_{i}|i\rangle\langle i|\otimes \bigotimes_{k= 1}^{N_{E}}\boldsymbol{\hat{\rho}}_{i}^{E^{k}} - \sum_{i_{1}}\sum_{i_{2}}...\sum_{i_{N_{E}}}\Bigg( \sum_{i}p_{i}|i\rangle\langle i|\otimes \bigotimes_{k= 1}^{N_{E}}\mathbf{\hat{M}}^{E^{k}}_{i_{k}}\boldsymbol{\hat{\rho}}_{i}^{E^{k}}\Big(\mathbf{\hat{M}}^{E^{k}}_{i_{k}}\Big)^{\dagger}\Bigg) \bigg\|_{1} =
\end{equation}
\begin{equation}
\frac{1}{2}\bigg\| \sum_{i}p_{i}|i\rangle\langle i|\otimes \bigotimes_{k= 1}^{N_{E}}\boldsymbol{\hat{\rho}}_{i}^{E^{k}} -  \sum_{i}p_{i}|i\rangle\langle i|\otimes \bigotimes_{k= 1}^{N_{E}}\mathbf{\hat{M}}^{E^{k}}_{i}\boldsymbol{\hat{\rho}}_{i}^{E^{k}}\Big(\mathbf{\hat{M}}^{E^{k}}_{i}\Big)^{\dagger}\bigg\|_{1} \leq
\end{equation}
\begin{equation}
\label{eqn:Canada3}
\frac{1}{2}\sum_{i}p_{i}\bigg\|  \bigotimes_{k= 1}^{N_{E}}\boldsymbol{\hat{\rho}}_{i}^{E^{k}} -  \bigotimes_{k= 1}^{N_{E}}\mathbf{\hat{M}}^{E^{k}}_{i}\boldsymbol{\hat{\rho}}_{i}^{E^{k}}\Big(\mathbf{\hat{M}}^{E^{k}}_{i}\Big)^{\dagger}\bigg\|_{1}
\end{equation}
To proceed we introduce the following lemma and we also take this opportunity to state two results that shall be used in the following.
\begin{definition2}{Telescopic inequality \cite{JKthree}:}
Let $\mathbf{\hat{A}}^{k}$ and $\mathbf{\hat{B}}^{k}$ be trace class operators for all $k$. Then,
\label{eqn:telescoping}
\begin{equation}
   \bigg\| \bigotimes_{k=1}^{N}\mathbf{\hat{A}}^{k}-\bigotimes_{k=1}^{N}\mathbf{\hat{B}}^{k}  \bigg\|_{1} \leq
\end{equation}
\begin{equation}
\sum_{j=1}^{N}\bigg(\prod_{k=1}^{j-1}\big\|\mathbf{\hat{A}}^{k}\big\|_{1}\bigg)\times \big\|\mathbf{\hat{A}}^{j}-\mathbf{\hat{B}}^{j}\big\|_{1}\times \bigg(\prod_{k=j+1}^{N}\big\|\mathbf{\hat{B}}^{k} 
 \big\|_{1}\bigg)
\end{equation}
\end{definition2}
In the following we use the probability error $p_{E}\big\{\{p_{i},\boldsymbol{\hat{\rho}}_{i}\}_{i=1}^{N}, \{\mathbf{\hat{M}}_{l}\big\}_{l=1}^{K} \big\}$ defined in Appendix \ref{app:QSD}.
\begin{definition4}{Montanaro Bound \cite{Montanaro}:}
\label{eqn:monta}
\begin{equation}
\min_{POVM}p_{E}\big\{\{p_{i},\boldsymbol{\hat{\rho}}_{i}\}_{i=1}^{N}, \{\mathbf{\hat{M}}_{l}\big\}_{l=1}^{K} \big\}\geq \frac{1}{2}\sum_{i}\sum_{j;j\neq i}p_{i}p_{j}F(\boldsymbol{\hat{\rho}}_{i},\boldsymbol{\hat{\rho}}_{j})
\end{equation}
\end{definition4}

\begin{definition4}{Knill and Barnum \cite{knill}}
\label{eqn:knill1}
\begin{equation}
\label{eqn:knill}
  \min_{POVM}p_{E}\big\{\{p_{i},\boldsymbol{\hat{\rho}}_{i}\}_{i=1}^{N}, \{\mathbf{\hat{M}}_{l}\big\}_{l=1}^{K} \big\} \leq \sum_{i}\sum_{j;j\neq i}\sqrt{p_{i}p_{j}}\sqrt{F\big(\boldsymbol{\hat{\rho}}_{i} , \boldsymbol{\hat{\rho}}_{j}\big)}  
\end{equation} 
\end{definition4}

Using Lemma \ref{eqn:telescoping}, (\ref{eqn:Canada3}) may be bounded as follows. 
\begin{equation}
\label{eqn:Canada4}
\frac{1}{2}\sum_{i}p_{i}\bigg\|  \bigotimes_{k= 1}^{N_{E}}\boldsymbol{\hat{\rho}}_{i}^{E^{k}} -  \bigotimes_{k= 1}^{N_{E}}\mathbf{\hat{M}}^{E^{k}}_{i}\boldsymbol{\hat{\rho}}_{i}^{E^{k}}\Big(\mathbf{\hat{M}}^{E^{k}}_{i}\Big)^{\dagger}\bigg\|_{1} \leq \frac{1}{2}\sum_{k=1}^{N_{E}}\sum_{i}p_{i}\bigg\|  \boldsymbol{\hat{\rho}}_{i}^{E^{k}} -  \mathbf{\hat{M}}^{E^{k}}_{i}\boldsymbol{\hat{\rho}}_{i}^{E^{k}}\Big(\mathbf{\hat{M}}^{E^{k}}_{i}\Big)^{\dagger}\bigg\|_{1} 
\end{equation}

 \section{Estimates for the Related Quantum State Discrimination Optimization Problem}
\;\;\; We claim that the distinguishability criterion $F\big(\boldsymbol{\hat{\rho}}_{i}^{E^{k}},\boldsymbol{\hat{\rho}}_{j}^{E^{k}} \big) =0$ ( $i\neq j$) for all $k$ is a sufficient condition for (\ref{eqn:Canada4}) to vanish. The sufficiency is immediate since each $\mathbf{\hat{E}}^{E^{k}}_{i}$ may be chosen to be a projector onto the domain of the respective $\boldsymbol{\hat{\rho}}_{i}^{E^{k}}$, meaning that $\mathbf{\hat{M}}^{E^{k}}_{i}\boldsymbol{\hat{\rho}}_{i}^{E^{k}}\Big(\mathbf{\hat{M}}^{E^{k}}_{i}\Big)^{\dagger} = \boldsymbol{\hat{\rho}}_{i}^{E^{k}}$ which in turn implies that 
\begin{equation}
\big\|  \boldsymbol{\hat{\rho}}_{i}^{E^{k}} -  \mathbf{\hat{M}}^{E^{k}}_{i}\boldsymbol{\hat{\rho}}_{i}^{E^{k}}\Big(\mathbf{\hat{M}}^{E^{k}}_{i}\Big)^{\dagger}\big\|_{1} = 0 
\end{equation}
The argument becomes more transparent when all of the $\boldsymbol{\hat{\rho}}_{i}^{E^{k}}$ are projectors. In this case we simply choose $\mathbf{\hat{E}}^{E^{k}}_{i}= \boldsymbol{\hat{\rho}}_{i}^{E^{k}}$.\\

The distinguishability condition $F\big(\boldsymbol{\hat{\rho}}_{i}^{E^{k}},\boldsymbol{\hat{\rho}}_{j}^{E^{k}} \big) =0$ ($i\neq j$) for all $k$ is of course an idealization; in practice there may be some error involved in the distinguishability measures $F\big(\boldsymbol{\hat{\rho}}_{i}^{E^{k}},\boldsymbol{\hat{\rho}}_{j}^{E^{k}} \big) = \varepsilon_{k}$ for all $k$. We then have to compute/estimate the sum over $i$ of trace norms in (\ref{eqn:Canada4}). Although Theorem \ref{eqn:knill1} implies
\begin{equation}
\label{eqn:morocco2}
\min_{POVM}\sum_{i}p_{i}Tr\big\{ \boldsymbol{\hat{\rho}}_{i}^{E^{k}} - \mathbf{\hat{M}}^{E^{k}}_{i}\boldsymbol{\hat{\rho}}_{i}^{E^{k}}\Big(\mathbf{\hat{M}}^{E^{k}}_{i}\Big)^{\dagger}\big\} \leq \sum_{i}\sum_{j;j\neq i}\sqrt{p_{i}p_{j}}\sqrt{F(\boldsymbol{\hat{\rho}}_{i}^{E^{k}},\boldsymbol{\hat{\rho}}_{i}^{E^{k}})}
\end{equation}
for all $k$, this does not aid us in minimize
\begin{equation}
\sum_{i}p_{i}\big\|  \boldsymbol{\hat{\rho}}_{i}^{E^{k}} - \mathbf{\hat{M}}^{E^{k}}_{i}\boldsymbol{\hat{\rho}}_{i}^{E^{k}}\Big(\mathbf{\hat{M}}^{E^{k}}_{i}\Big)^{\dagger}\big\|_{1}
\end{equation}
It is clear by (\ref{eqn:morocco2})  that when the fidelities $F\big(\boldsymbol{\hat{\rho}}_{i}^{E^{k}},\boldsymbol{\hat{\rho}}_{j}^{E^{k}} \big)$ ($i\neq j$) are small, then for all $k$ the local QSD error 
\begin{equation}
\min_{POVM}Tr\{ \boldsymbol{\hat{\rho}}_{i}^{E^{k}} - \mathbf{\hat{M}}^{E^{k}}_{i}\boldsymbol{\hat{\rho}}_{i}^{E^{k}}\Big(\mathbf{\hat{M}}^{E^{k}}_{i}\Big)^{\dagger}\}
\end{equation}
is also small. To show that a similar argument holds for the right-hand side of the inequality $(\ref{eqn:Canada4})$ we will prove in the following section a bound for (\ref{eqn:Canada3}) that will depend only on fidelities between the density operators $\{\boldsymbol{\hat{\rho}}_{i}^{E^{k}}\}_{i}$, and vanish as the fidelities $F\big(\boldsymbol{\hat{\rho}}_{i}^{E^{k}},\boldsymbol{\hat{\rho}}_{j}^{E^{k}} \big)$ ($i\neq j$)
 decay to zero for all $k$. We will first show this for the case where the $\boldsymbol{\hat{\rho}}_{i}^{E^{k}}$ are pure states and then generalize our results to the case where these operators are finite mixtures.\\
 
In the present section, we simplify our notational conventions, since we do not need the superscripts on the density operators used in the previous section. Consider the mixed state $\sum_{i=1}^{N}p_{i}\boldsymbol{\hat{\rho}}_{i}$, where $\sum_{i=1}^N p_i = 1$ and the $\boldsymbol{\hat{\rho}}_i$ are pure states in a Hilbert space of dimension greater than $N$, i.e. one-dimensional projections $\big|\psi_{i}\big\rangle\big\langle \psi_{i}\big|$, where $\big\{\big|\psi_{i}\big\rangle\big\}_{i=1}^{N}$ are normalized vectors. Assuming that $ \big|\psi_{i}\big\rangle$ are linearly independent, we may use the well-known Gram-Schmidt procedure to define the associated orthonormal set. 

\begin{definition}{Gram-Schmidt Procedure:}
\label{eqn:gramschmidtpro}
Assume that the set $\big\{\big|\psi\big\rangle_{i}\big\}_{i=1}^{N}$, of vectors in some vector space $V$, is a linearly independent set. Then the following construction yields an orthonormal set. 
\begin{equation}
\big|\phi_{1}\big\rangle = \big|\psi_{1}\big\rangle
\end{equation}
\begin{equation}
\big|\phi_{2}\big\rangle = \frac{1}{\alpha_{2}}\bigg(\big|\psi_{2}\big\rangle - \big\langle \phi_{1}\big|\psi_{2}\big\rangle\big|\phi_{1}\big\rangle\bigg)
\end{equation}
$$\vdots$$
\begin{equation}
\big|\phi_{N}\big\rangle = \frac{1}{\alpha_{N}}\bigg(\big|\psi_{N}\big\rangle - \sum_{k=1}^{N-1}\big\langle \phi_{k}\big|\psi_{N}\big\rangle\big|\phi_{k}\big\rangle \bigg)
\end{equation}
Here $\alpha_{i}: = \Big\|\big|\psi_{i}\big\rangle - \sum_{k=1}^{i-1}\Big\langle \phi_{k}\big|\psi_{i}\big\rangle\big|\phi_{k}\big\rangle\big\| = \sqrt{1-\sum_{k=1}^{i-1}\big|\big\langle\phi_{k}\big|\psi_{i}\big\rangle\big|^{2}}$ for $i>1$ and $\alpha_{1} = 1$ are the respective normalizing constants. 
We have $Span\Big\{\big\{\big|\psi_{i}\big\rangle\big\}_{i=1}^{N}\Big\} = Span\Big\{\big\{\big|\phi_{i}\big\rangle\big\}_{i=1}^{N}\Big\}$. 
\end{definition}
The orthonormal set $\big\{ \big|\phi_{i}\big\rangle\big\}_{i=1}^{N}$ may be used for the construction of a $PVM$, namely
\begin{equation}
\label{eqn:2.122}
\Big\{\big|\phi_{i}\big\rangle \big\langle \phi_{i}\big|\Big\}_{i=1}^{N}\cup\bigg\{\mathbb{I}-\sum_{i=1}^{N}\big|\phi_{i}\big\rangle\big\langle \phi_{i}\big|\bigg\}
\end{equation}
which we will use to estimate $\min_{POVM}\sum_{i=1}^{N}p_{i}\big\|\boldsymbol{\hat{\rho}}_{i} - \mathbf{\hat{M}}_{i}\boldsymbol{\hat{\rho}}_{i}\mathbf{\hat{M}}^{\dagger}_{i}\big\|_{1}$: 
\begin{equation}
\min_{POVM}\sum_{i=1}^{N}p_{i}Tr\Big\{\boldsymbol{\hat{\rho}}_{i} - \mathbf{\hat{M}}_{i}\boldsymbol{\hat{\rho}}_{i}\mathbf{\hat{M}}^{\dagger}_{i}\Big\}\leq\min_{POVM}\sum_{i=1}^{N}p_{i}\big\|\boldsymbol{\hat{\rho}}_{i} - \mathbf{\hat{M}}_{i}\boldsymbol{\hat{\rho}}_{i}\mathbf{\hat{M}}^{\dagger}_{i}\big\|_{1}\leq
\end{equation}
\begin{equation}
\sum_{i=1}^{N}p_{i}\Big\| \boldsymbol{\hat{\rho}}_{i}- \big|\phi_{i}\big\rangle\big\langle \phi_{i}\big| \boldsymbol{\hat{\rho}}_{i}\big|\phi_{i}\big\rangle\big\langle \phi_{i}\big|\Big\|_{1}
\end{equation}
for a judiciously chosen PVM $\big\{\big|\phi_{i}\big\rangle\big\langle \phi_{i}\big|\big\}_{i}$.

\begin{definition2}
\label{eqn:lemmaboundwithgramschmidt}Let $\boldsymbol{\hat{\rho}}_{i}$ and $\big|\phi_{i}\big\rangle$ be defined as above; also let $i>1$, then
\begin{equation}
\big\| \boldsymbol{\hat{\rho}}_{i}- |\phi_{i}\rangle\langle \phi_{i}| \boldsymbol{\hat{\rho}}_{i}|\phi_{i}\rangle\langle \phi_{i}\big\|_{1}\leq 2\sum_{k=1}^{i-1}|\langle \phi_{k}|\psi_{i}\rangle|
\end{equation}
\end{definition2}
\begin{proof}
\begin{equation}
\Big\| \boldsymbol{\hat{\rho}}_{i}- \big|\phi_{i}\big\rangle\big\langle \phi_{i}\big| \boldsymbol{\hat{\rho}}_{i}\big|\phi_{i}\big\rangle\big\langle \phi_{i}\big|\Big\|_{1} = \Big\| \big|\psi_{i}\big\rangle\big\langle\psi_{i}\big|\boldsymbol{\hat{\rho}}_{i}\big|\psi_{i}\big\rangle\big\langle\psi_{i}\big|- \big|\phi_{i}\big\rangle\big\langle \phi_{i}\big| \boldsymbol{\hat{\rho}}_{i}\big|\phi_{i}\big\rangle\big\langle \phi_{i}\big|\Big\|_{1}  =
\end{equation}
\begin{equation}
\Big\| \Big(\big|\psi_{i}\big\rangle\big\langle \psi_{i}\big| - \big|\phi_{i}\big\rangle\big\langle \phi_{i}\big|\Big)\boldsymbol{\hat{\rho}}_{i}\big|\psi_{i}\big\rangle\big\langle \psi_{i}\big|+  \big|\phi_{i}\big\rangle\big\langle \phi_{i}\big|\boldsymbol{\hat{\rho}}_{i}\Big(\big|\psi_{i}\big\rangle\big\langle \psi_{i}\big| - \big|\phi_{i}\big\rangle\big\langle \phi_{i}\big|\Big)\Big\|_{1} \leq
\end{equation}

\begin{equation}
\Big\| \Big(\big|\psi_{i}\big\rangle\big\langle \psi_{i}\big| - \big|\phi_{i}\big\rangle\big\langle \phi_{i}\big|\Big)\boldsymbol{\hat{\rho}}_{i}\big|\psi_{i}\big\rangle\big\langle \psi_{i}\big|\Big\|_{1}+  \Big|\big|\phi_{i}\big\rangle\big\langle \phi_{i}\big|\boldsymbol{\hat{\rho}}_{i}\Big(\big|\psi_{i}\big\rangle\big\langle \psi_{i}\big| - \big|\phi_{i}\big\rangle\big\langle \phi_{i}\big|\Big)\Big\|_{1} \leq
\end{equation}
\begin{equation}
\Big\| \big|\psi_{i}\big\rangle\big\langle \psi_{i}\big| - \big|\phi_{i}\big\rangle\big\langle \phi_{i}\big|\Big\|_{1}\Big\|\boldsymbol{\hat{\rho}}_{i}\big|\psi_{i}\big\rangle\big\langle \psi_{i}\big|\Big\|_{1} + \Big\|\big|\phi_{i}\big\rangle\big\langle \phi_{i}\big|\boldsymbol{\hat{\rho}}_{i}\Big\|_{1}\Big\|\big|\psi_{i}\big\rangle\big\langle \psi_{i}\big| - \big|\phi_{i}\big\rangle\big\langle \phi_{i}\big|\Big\|_{1} =
\end{equation}
\begin{equation}
\Big\| \big|\psi_{i}\big\rangle\big\langle \psi_{i}\big| - \big|\phi_{i}\big\rangle\big\langle \phi_{i}\big|\Big\|_{1}\Bigg(\Big\|\boldsymbol{\hat{\rho}}_{i}\big|\psi_{i}\big\rangle\big\langle \psi_{i}\big|\Big\|_{1} + \Big\|\big|\phi_{i}\big\rangle\big\langle \phi_{i}\big|\boldsymbol{\hat{\rho}}_{i}\Big\|_{1}\Bigg)\leq
\end{equation}
\begin{equation}
\Big\| \big|\psi_{i}\big\rangle\big\langle \psi_{i}\big| - \big|\phi_{i}\big\rangle\big\langle \phi_{i}\big|\Big\|_{1}\Bigg(\big\|\boldsymbol{\hat{\rho}}_{i}\big\|_{1}\Big\|\big|\psi_{i}\big\rangle\big\langle \psi_{i}\big|\Big\|_{1} + \Big\|\big|\phi_{i}\big\rangle\big\langle \phi_{i}\big|\Big\|_{1}\big\|\boldsymbol{\hat{\rho}}_{i}\big\|_{1}\Bigg)\leq
\end{equation}
\begin{equation}
2\Big\| \big|\psi_{i}\big\rangle\big\langle \psi_{i}\big| - \big|\phi_{i}\big\rangle\big\langle \phi_{i}\big|\Big\|_{1} = 2\sqrt{1-|\big\langle \psi_{i}\big|\phi_{i}\big\rangle|^{2}} =
\end{equation}
\begin{equation}
 2\sqrt{1 - \Bigg|\frac{1}{\alpha_{i}}\bigg( 1-\sum_{k=1}^{i-1}|\big\langle \phi_{k}\big|\psi_{i}\big\rangle|^{2}\bigg)\Bigg|^{2}} =  2\sqrt{1 - \Bigg|\frac{\Big(1-\sum_{k=1}^{i-1}|\big\langle \phi_{k}\big|\psi_{i}\big\rangle|^{2}\Big)}{\sqrt{\Big(1-\sum_{k=1}^{i-1}|\big\langle \phi_{k}\big|\psi_{i}\big\rangle|^{2}\Big)}}\Bigg|^{2}}
\end{equation}
\begin{equation}
 = 2\sqrt{1 - 1+\sum_{k=1}^{i-1}|\big\langle \phi_{k}\big|\psi_{i}\big\rangle|^{2}} = 2\sqrt{\sum_{k=1}^{i-1}|\big\langle \phi_{k}\big|\psi_{i}\big\rangle|^{2}} \leq 2\sum_{k=1}^{i-1}|\big\langle \phi_{k}\big|\psi_{i}\big\rangle|
\end{equation}
\end{proof}
The term $\sum_{k=1}^{i-1}\big|\big\langle \phi_{k}\big|\psi_{i}\big\rangle|$ may be understood in terms of the related \emph{Gram Determinant}. We present this as a lemma. 
\begin{definition2}{(See \cite{gogo})}
\label{eqn:lemmagramschmidtmatrix}
\begin{equation}
\big|\phi_{j}\big\rangle  = \frac{1}{\sqrt{D_{j-1}D_{j}}} \mathbf{det}
\begin{pmatrix}
\big\langle\psi_{1}\big|\psi_{1}\big\rangle  & \big\langle \psi_{1}\big|\psi_{2}\big\rangle & \dots & \big\langle\psi_{1}\big|\psi_{j}\big\rangle\\

\big\langle \psi_{2}\big|\psi_{1}\big\rangle & \big\langle \psi_{2}\big|\psi_{2}\big\rangle& \dots & \big\langle \psi_{2}\big|\psi_{j}\big\rangle\\

\vdots & \vdots & \ddots & \vdots\\

\big\langle \psi_{j-1}\big|\psi_{1}\big\rangle & \big\langle \psi_{j-1}\big|\psi_{2}\big\rangle & \dots & \big\langle \psi_{j-1}\big|\psi_{j}\big\rangle\\

\big|\psi_{1}\big\rangle & \big|\psi_{2}\big\rangle & \dots &\big|\psi_{j}\big\rangle\\
\end{pmatrix}
\end{equation}
where 
\begin{equation}
\hspace*{1.85cm}D_{j} : = \mathbf{det}\begin{pmatrix}
\langle \psi_{1}|\psi_{1}\rangle  & \langle \psi_{1}|\psi_{2}\rangle & \dots & \langle\psi_{1}|\psi_{j}\rangle\\

\langle \psi_{2}|\psi_{1}\rangle & \langle \psi_{2}|\psi_{2}\rangle& \dots & \langle \psi_{2}|\psi_{j}\rangle\\

\vdots & \vdots & \ddots & \vdots\\

\langle \psi_{j}|\psi_{1}\rangle & \langle \psi_{j}|\psi_{2}\rangle & \dots &\langle \psi_{j}|\psi_{j}\rangle\\
\end{pmatrix}
\end{equation}
where $D_{0}:= 1$.
\end{definition2}
In determinant form, $\big\langle \psi_{i}\big|\phi_{k}\big\rangle$ may now be written as follows. 
\begin{center}
\begin{equation}
\label{eqn:thematrixgram}
\big\langle \psi_{i}\big|\phi_{k}\big\rangle = 
\frac{1}{\sqrt{D_{k-1}D_{k}}}\mathbf{det}\begin{pmatrix}
\big\langle \psi_{1}\big|\psi_{1}\big\rangle  & \big\langle \psi_{1}\big|\psi_{2}\big\rangle & \dots & \big\langle\psi_{1}\big|\psi_{k}\big\rangle\\

\big\langle \psi_{2}\big|\psi_{1}\big\rangle & \big\langle \psi_{2}\big|\psi_{2}\big\rangle& \dots & \big\langle \psi_{2}\big|\psi_{k}\big\rangle\\

\vdots & \vdots & \ddots & \vdots\\

\big\langle \psi_{k-1}\big|\psi_{1}\big\rangle & \big\langle \psi_{k-1}\big|\psi_{2}\big\rangle & \dots & \big\langle \psi_{k-1}\big|\psi_{k}\big\rangle\\

\big\langle \psi_{i}\big|\psi_{1}\big\rangle & \big\langle \psi_{i}\big|\psi_{2}\big\rangle & \dots &\big\langle \psi_{i}\big|\psi_{k}\big\rangle\\
\end{pmatrix}
\end{equation}
\end{center}
The power of viewing the states $\big|\phi_{i}\big\rangle$ in the determinant form is that now we need only compute inner products between elements of the set $\big\{\big|\psi_{i}\big\rangle\big\}_{i=1}^{N}$ in order to estimate the solution of the $PVM$ (\ref{eqn:2.122}) via an approximate solution to (\ref{eqn:Canada3}) problem with a particular PVM, i.e. 
\begin{equation}
\min_{PVM} \sum_{i=1}^{N}p_{i}\big\|\boldsymbol{\hat{\rho}}_{i} - \mathbf{\hat{P}}_{i}\boldsymbol{\hat{\rho}}_{i}\mathbf{\hat{P}}_{i}\big\|_{1}
\end{equation}
Recall that the states $\big\{\big|\psi_{i}\big\rangle\big\}_{i=1}^{N}$ are normalized and let us consider the case where $\big\langle \psi_{i}\big|\psi_{j}\big\rangle = \varepsilon_{ij}$ for all $i \neq j \in \{1,..., N\}$, where $\varepsilon_{ij}$ are complex numbers satisfying $|\varepsilon_{ij}| \leq \delta$ for all $i\neq j \in \{1,..,N\}$, where $\delta$ is small. Since, under this assumption, all entries of the last column of the matrix (\ref{eqn:thematrixgram}) are small, this would also imply that $ \Big\| \boldsymbol{\hat{\rho}}_{i}- |\phi_{i}\rangle\langle \phi_{i}| \boldsymbol{\hat{\rho}}_{i}|\phi_{i}\rangle\langle \phi_{i}\Big\|_{1}$ is small for all $i$, thanks to Lemma \ref{eqn:lemmaboundwithgramschmidt}.\\

The above estimates imply the following theorem. 
\begin{definition4}
\label{eqn:theoremgram}
Consider a mixed state of the form $\sum_{i=1}^{N}p_{i}\boldsymbol{\hat{\rho}}_{i}$,\:$\sum_{i=1}^{N}p_{i}=1$, where $\boldsymbol{\hat{\rho}}_{i}:=\big|\psi_{i}\big\rangle\big\langle \psi_{i}\big|$ are pure states acting on a \emph{Hilbert} space of dimension greater than $N$. Furthermore, assume that the states $\{\big|\psi_{i}\big\rangle\}_{i}$ are linearly independent. Then
\begin{equation}
\label{eqn:morocco3}
\min_{POVM} \sum_{i=1}^{N}p_{i}\bigg\|\boldsymbol{\hat{\rho}}_{i} -\mathbf{\hat{M}}_{i}\boldsymbol{\hat{\rho}}_{i}\mathbf{\hat{M}}^{\dagger}_{i}\bigg\|_{1}\leq \sum_{i=2}^{N}p_{i}\sum_{k=1}^{i-1}\Big|\frac{M_{k,i}}{D_{k-1}D_{k}}\bigg|
\end{equation}
where 
\begin{equation}
\label{eqn:theend7}
 M_{k,i}:= 2\mathbf{det}\begin{pmatrix}
\langle \psi_{1}|\psi_{1}\rangle  & \langle \psi_{1}|\psi_{2}\rangle & \dots & \langle\psi_{1}|\psi_{k}\rangle\\

\langle \psi_{2}|\psi_{1}\rangle & \langle \psi_{2}|\psi_{2}\rangle& \dots & \langle \psi_{2}|\psi_{k}\rangle\\

\vdots & \vdots & \ddots & \vdots\\

\langle \psi_{k-1}|\psi_{1}\rangle & \langle \psi_{k-1}|\psi_{2}\rangle& \dots & \langle \psi_{k-1}|\psi_{k}\rangle\\

\langle \psi_{i}|\psi_{1}\rangle & \langle \psi_{i}|\psi_{2}\rangle & \dots &\langle \psi_{i}|\psi_{k}\rangle\\
\end{pmatrix} 
\end{equation}
\begin{equation}
\label{eqn:morocco4}
\hspace*{0.8cm} D_{k} : = \mathbf{det}\begin{pmatrix}
\langle \psi_{1}|\psi_{1}\rangle  & \langle \psi_{1}|\psi_{2}\rangle & \dots & \langle\psi_{1}|\psi_{k}\rangle\\

\langle \psi_{2}|\psi_{1}\rangle & \langle \psi_{2}|\psi_{2}\rangle& \dots & \langle \psi_{2}|\psi_{k}\rangle\\

\vdots & \vdots & \ddots & \vdots\\

\langle \psi_{k}|\psi_{1}\rangle & \langle \psi_{k}|\psi_{2}\rangle & \dots &\langle \psi_{k}|\psi_{k}\rangle\\
\end{pmatrix}
\end{equation}
\end{definition4}

\begin{proof}
The proof follows directly from Lemma \ref{eqn:lemmagramschmidtmatrix} and Lemma \ref{eqn:lemmaboundwithgramschmidt}, 
and the fact that for $i=1$ the corresponding projector is simply $\big|\psi_{i}\big\rangle\big\langle \psi_{i}\big|$ making the $i = 1$ term zero. 
\end{proof}
\section{Dynamical Monitoring for discrete variables}\label{sec:monitoring}
\;\;\; We have been studying the properties of SBS states as introduced in Definition \ref{eqn:turk4}. In the present section, we will focus on the asymptotic behavior of the time-dependent density operator $\boldsymbol{\hat{\rho}}_{t}$ evolving non-unitarily. If we knew \emph{a priori} that a certain type of multipartite quantum-mechanical system behaves objectively as in Definitions \ref{def:objectivity} and \ref{eqn:turk4}, then the states of this system should be SBS-asymptotic as $t\rightarrow \infty$, in the following sense.
\begin{definition}{SBS-asymptotic:}
We say that time-dependent density operator $\boldsymbol{\hat{\rho}}_{t}$ is SBS-asymptotic whenever
\begin{equation}
\lim_{t\rightarrow\infty}\min_{SBS}\big\|\boldsymbol{\hat{\rho}}_{t}-\boldsymbol{\hat{\rho}}_{SBS}\big\|_{1}\rightarrow 0
\end{equation}
where the minimum is over SBS states $\boldsymbol{\hat{\rho}}_{SBS}$.
\end{definition}
Time evolution may in general be described by an arbitrary time-dependent quantum map. We will focus on the quantum maps obtained by partial trace from unitary evolution corresponding to a Hamiltonian specified below.

\subsection{Quantum-Measurement Limit}
\;\;\; The principal models studied in SBS literature \cite{JKone}\cite{JKtwo}\cite{JKthree} are of the \emph{quantum-measurement} limit type, meaning SBS that arise as $t$ converges to infinity from dynamics generated by Hamiltonians in which the interaction term between the system $S$ and the environment $E$ dominates, i.e. $\mathbf{\hat{H}}_{tot} = \mathbf{\hat{H}}_{int}+\mathbf{\hat{H}}_{E}+\mathbf{\hat{H}}_{S} \approx \mathbf{\hat{H}}_{int}$ (here \emph{tot} means total and \emph{int} means interaction terms). Such an approximation is valid when the system and the environments evolve on a time scale that is much larger than that corresponding to the interactive dynamics; these types of dynamics are central to the theory of \emph{quantum decoherence} \cite{schloss2}. In this work, we will furthermore narrow our focus to interaction Hamiltonians of the von Neumann type \cite{zeh}
\begin{equation}
\label{eqn:intham}
\mathbf{\hat{H}}_{int} =  \mathbf{\hat{X}}\otimes\sum_{k=1}^{N}g_{k}\mathbf{\hat{B}}_{k}
\end{equation}
The corresponding time evolution operator is thus
\begin{equation}
\label{eqn:timeev11}
\mathbf{\hat{U}}_{t} = e^{-it\mathbf{\hat{X}}\otimes\sum_{k=1}^{N}g_{k}\mathbf{\hat{B}}_{k}}.
\end{equation}

The theory of SBS for discrete variables focuses on the case where Hilbert space of the system $S$ is finite-dimensional \cite{JKone}\cite{JKtwo}\cite{JKthree}. In this case, the self-adjoint operator $\mathbf{\hat{X}}$ has purely discrete spectrum. Let $\{\big|i\big\rangle\}_{i=1}^{d_{S}}$ be the set of eigenvectors of $\mathbf{\hat{X}}$ with $x_{i}$ being the corresponding eigenvalues. $\mathbf{\hat{B}}_{k}$ is assumed to be an arbitrary self-adjoint operator. We shall see that the spectral properties of the operator $\mathbf{\hat{B}}$ determine whether or not the multipartite states converge to an SBS state. \\

\subsection{Partial Tracing}
\;\;\; We consider a quantum system interacting with $N$ macroscopic environments. We assume that the joint initial state has the product form: 
\begin{equation}
\label{eqn:snake}
\boldsymbol{\hat{\rho}}=\boldsymbol{\hat{\rho}}_{S}\otimes\bigotimes_{k=1}^{N}\boldsymbol{\hat{\rho}}^{E^{k}}
\end{equation}
We evolve our total initial state using the evolution operator (\ref{eqn:timeev11}). 
\begin{equation}
\label{eqn:zelda10}
\boldsymbol{\hat{\rho}}_{t} = \bigg(e^{-it\mathbf{\hat{X}}\otimes\sum_{k=1}^{N}g_{k}\mathbf{\hat{B}}_{k}}\bigg)\boldsymbol{\hat{\rho}}_{S}\otimes\bigotimes_{k=1}^{N}\boldsymbol{\hat{\rho}}^{E^{k}}\bigg(e^{it\mathbf{\hat{X}}\otimes\sum_{k=1}^{N}g_{k}\mathbf{\hat{B}}_{k}}\bigg).
\end{equation}
To study the state of the subsystem formed by the system $S$ and the first $N_{E}$ environments, we take the partial trace of the time-evolved density operator over the remaining $M_{E}:=N-N_{E}$ environments. The result is, 
\begin{equation}
\label{eqn:decmode}
\sum_{i,j=1}^{d_{S}}\sigma_{i,j}\Gamma(i,j,t)|i\rangle\langle j| \otimes\bigotimes_{k=1}^{N_{E}} \boldsymbol{\hat{\rho}}_{x_{i},x_{j}}^{E^{k},t}
\end{equation}
where, again, $\{\big|i\big\rangle\}_{i=1}^{d_{s}}$ are the eigenvectors of $\mathbf{\hat{X}}$, with corresponding eigenvalues $\{x_{i}\}_{i=1}^{d_{S}}$. Here, we have use the following notation
\begin{equation}
\boldsymbol{\hat{\rho}}_{x,y}^{E^{k},t}:= e^{-itxg_{k}\mathbf{\hat{B}}_{k}}\boldsymbol{\hat{\rho}}^{E^{k}}e^{ityg_{k}\mathbf{\hat{B}}_{k}} \: (k  =1,2,..., N_{E}) 
\end{equation}
\begin{equation}
\label{eqn:purepreserve}
\boldsymbol{\hat{\rho}}_{x}^{E^{k},t}:= e^{-itxg_{k}\mathbf{\hat{B}}_{k}}\boldsymbol{\hat{\rho}}^{E^{k}}e^{itxg_{k}\mathbf{\hat{B}}_{k}} \: (k = 1,2,...,N_{E}). 
\end{equation}
\begin{equation}
\sigma_{i,j} := \langle i|\boldsymbol{\hat{\rho}}_{S}|j\rangle 
\end{equation}
\begin{equation}
\gamma_{i,j}^{k}(t):= Tr\big\{ \boldsymbol{\hat{\rho}}_{x_{i},x_{j}}^{E^{k},t} \big\}
\end{equation}
\begin{equation}
\label{eqn:thegamm}
\Gamma(i,j,,t):=\prod_{n=N_{E}+1}^{N}\gamma_{i,j}^{n}(t)
\end{equation}
We may also write (\ref{eqn:decmode}) as 
\begin{equation}
\label{eqn:decmodechan}
\Lambda_{t}\Big(\boldsymbol{\hat{\rho}}_{S}\otimes\bigotimes_{k=1}^{N_{E}}\boldsymbol{\hat{\rho}}^{E^{k}}\Big):=
\sum_{i,j=1}^{d_{S}}\sigma_{i,j}\Gamma(i,j,t)|i\rangle\langle j| \otimes\bigotimes_{k=1}^{N_{E}} \boldsymbol{\hat{\rho}}_{x_{i},x_{j}}^{E^{k},t}
\end{equation}
where $\Lambda$ is a quantum channel (see \cite{Nielsen} for a discussion of quantum channels) defined as follows.
\begin{equation}
\label{eqn:morocco7}
\Lambda_{t}\big(\boldsymbol{\hat{\rho}}\big):=\mathscr{U}_{t}\circ\mathscr{E}_{t}\big(\boldsymbol{\hat{\rho}}\big)
\end{equation}
where 
\begin{equation}
\mathscr{U}_{t}\big(\mathbf{\hat{A}}\big):= e^{-it\mathbf{\hat{X}}\otimes  \sum_{k=1}^{N_{E}}g_{k}\mathbf{\hat{B}}_{k}}\big( \mathbf{\hat{A}}\big)e^{it\mathbf{\hat{X}}\otimes\sum_{k=1}^{N_{E}}g_{k}\mathbf{\hat{B}}_{k}}
\end{equation}
acts non-trivially in $\mathcal{S}\big(\mathscr{H}_{S}\otimes \bigotimes_{k=1}^{N_{E}}\mathscr{H}_{k}\big)$ 
and  $\mathscr{E}_{t}$ acts non-trivially only in $\mathcal{S}\big(\mathscr{H}_{S})$ as follows. 
\begin{equation}
\label{eqn:decfinitecase}
\mathscr{E}_{t}(\mathbf{\hat{C}}):= \sum_{i,j=1}^{d_{S}}\langle i|\mathbf{\hat{C}}|j\rangle\Gamma(i,j,t)|i\rangle\langle j|
\end{equation}
In words, the trace-preserving quantum map $\Lambda_{t}$ is a composition of two trace-preserving quantum maps $\mathscr{U}_{t}$ and $\mathscr{E}_{t}$: a unitary map acting on $S$ and the environmental degrees of freedom that were not traced out, and a non-unitary map acting locally in $S$.

\section{Monitoring the Process Information Broadcasting by the System} 
\;\;\; In \cite{JKthree} the goal was to show that (\ref{eqn:decmode}) is SBS-asymptotic. For a given $t\geq 0$, one can create an SBS state approximating (\ref{eqn:decmode}) in the following way. We first restrict the sum of (\ref{eqn:decmode}) to the diagonal terms---the terms with $i=j$. We will use the index "diag" to indicate this: 
\begin{equation}
\label{eqn:dig}
\boldsymbol{\hat{\rho}}_{diag,t}= \sum_{i=1}^{d_{S}}\sigma_{i}|i\rangle\langle i| \otimes\bigotimes_{k=1}^{N_{E}} \boldsymbol{\hat{\rho}}_{x_{i}}^{E^{k},t} 
\end{equation}
The next step is to choose for every $t$ a PVM acting on the space $\mathcal{S}\big(\mathscr{H}_{S}\otimes\bigotimes_{k=1}^{N_{E}} \mathscr{H}_{E^{k}}\big)$ ( Note that for the case considered in \cite{JKthree}, $dim(\mathscr{H}_{S}) = d_S < \infty$ and $dim(\mathscr{H}_{E^{k}})= d_{E^{k}} < \infty $ for all $k$). To define such a PVM, we use the eigenbasis of the operator $\mathbf{\hat{X}}$: the elements of the $PVM$ are $\big|i\big\rangle \big\langle i\big|\otimes\bigotimes_{k=1}^{N_{E}} \mathbf{\hat{P}}^{E^{k},t}_{j}$ where the $\big\{\big|i\big\rangle\big\langle i\big| \big\}_{i=1}^{d_{S}}$ and $\big\{\mathbf{\hat{P}}^{E^{k},t}_{j}\big\}_{j=1}^{d_{S}}\cup\big\{\mathbb{I}-\sum_{i=1}^{d_{S}}\mathbf{\hat{P}}^{E^{k},t}_{i}\big\}$ resolve the identity operators in $\mathscr{H}_{S}$ and $\mathscr{H}_{E_{k}}$ respectively. These PVMs are then used to approximate the operator (\ref{eqn:decmode}) by an SBS state:
\begin{equation}
\label{eqn:sbsfinally}
\boldsymbol{\hat{\rho}}_{SBS,t} :=  \frac{1}{\mathscr{N}}\sum_{j=1}^{d_{S}}\Bigg(\big|j\big\rangle \big\langle j\big|\otimes\bigotimes_{k=1}^{N_{E}} \boldsymbol{P}^{E^{k},t}_{j}\Bigg)\boldsymbol{\hat{\rho}}_{diag,t}\Bigg(\big|j\big\rangle \big\langle j\big|\otimes\bigotimes_{k=1}^{N_{E}} \mathbf{\hat{P}}^{E^{k},t}_{j} \Bigg)= 
\end{equation}
\begin{equation}
\label{eqn:SBS20canada}
\sum_{i=1}^{d_{S}}\Tilde{\sigma}_{i}|i\rangle\langle i| \otimes\bigotimes_{k=1}^{N_{E}} \Bigg(\mathbf{\hat{P}}_{i}^{E^{k},t}\boldsymbol{\hat{\rho}}_{x_{i}}^{E^{k},t}\mathbf{\hat{P}}_{i}^{E^{k},t}\Bigg).
\end{equation}
Here $\mathscr{N}$ is a normalizing constant and $\Tilde{\sigma}_{i}: = \frac{\sigma_{i}}{\mathscr{N}}$.
One can verify that the operator (\ref{eqn:SBS20canada}) is indeed an SBS state as defined in Definition \ref{eqn:turk4}. If the trace distance between (\ref{eqn:decmode}) and  (\ref{eqn:SBS20canada}) approaches 0 as $t \to\infty$, it is indeed the case that (\ref{eqn:decmode}) is SBS-asymptotic. Namely, if
\begin{equation}
\label{eqn:dist}
\min_{PVM}\big\|  \boldsymbol{\hat{\rho}}_{t} - \boldsymbol{\hat{\rho}}_{SBS,t}   \big\|_{1} \to 0 \:\: \text{as} \:\: t\to \infty 
\end{equation}
where for each $t$ the minimization is taken over all projector-valued-measures \begin{equation}
 \big\{\mathbf{\hat{P}}^{E^{k},t}_{i}\}_{i=1}^{d_{S}}\cup\{\mathbb{I}-\sum_{i=1}^{d_{S}}\mathbf{\hat{P}}^{E^{k},t}_{i}\}
 \end{equation}
then
\begin{equation}
\min_{SBS}\big\|  \boldsymbol{\hat{\rho}}_{t} - \boldsymbol{\hat{\rho}}_{SBS,t}   \big\|_{1} \to 0 \:\: \text{as} \:\: t\to \infty 
\end{equation}
since
\begin{equation}
 \label{eqn:inq127}
 \min_{SBS}\big\|  \boldsymbol{\hat{\rho}}_{t} - \boldsymbol{\hat{\rho}}_{SBS,t}   \big\|_{1}\leq \min_{PVM}\big\|  \boldsymbol{\hat{\rho}}_{t} - \boldsymbol{\hat{\rho}}_{SBS,t}   \big\|_{1}
\end{equation}
An attempt is made in \cite{JKthree} to prove (\ref{eqn:dist}) but the argument provided there is incomplete. In what follows we discuss the bounds presented in \cite{JKthree}, as well as propose and prove an alternative bound for the trace distance in (\ref{eqn:dist}).\\

In \cite{JKthree}, the following bound is conjectured for the trace distance in (\ref{eqn:dist}).
 \begin{equation}
 \label{eqn:theirs2}
\frac{1}{2}\min_{PVM}\big\|  \boldsymbol{\hat{\rho}}_{t} - \boldsymbol{\hat{\rho}}_{SBS,t}  \big\|_{1} \leq \Gamma(t) +\sum_{i}\sum_{j;j\neq i}\sqrt{\sigma_{i}\sigma_{j}}\sum_{k=1}^{N_{E}}F\big(\boldsymbol{\hat{\rho}}_{x_{i}}^{E^{k},t}, \boldsymbol{\hat{\rho}}_{x_{j}}^{E^{k},t}  \big)
\end{equation} 
where now, $\Gamma(t) :=  \sum_{i}\sum_{j;j\neq i}|\sigma_{i,j}| \prod_{k=N_{E}+1}^{N}| \gamma_{i,j}^{k}(t)|$, and again $\gamma_{i,j}^{k}(t)=Tr\big[ \boldsymbol{\hat{\rho}}_{x_{i},x_{j}}^{E^{k},t}  \big]$, $\sigma_{i,j} := \langle i|\boldsymbol{\hat{\rho}}_{S_{0}}|j\rangle$ . This result would allow us to estimate the minimum on the left-hand side of (\ref{eqn:inq127}), using the asymptotic properties of $\Gamma(t)$ and the fidelity factors in (\ref{eqn:theirs2}). This bound would in turn provide a way to estimate $\frac{1}{2}\min_{POVM}\big\|  \boldsymbol{\hat{\rho}}_{t} - \boldsymbol{\hat{\rho}}_{SBS,t}  \big\|_{1}$. As (\ref{eqn:theirs2}) is currently not known to be true, we will not be using it. Instead, we will use Theorem \ref{eqn:theoremgram} proven in the previous section.

\subsection{A New Bound for the Trace Distance of a Multipartite State and an Approximating SBS State}
\;\;\;In what follows we use an unnormalized version of (\ref{eqn:sbsfinally}):   $\boldsymbol{\hat{\rho}}_{UN,t}:= \mathscr{N}\boldsymbol{\hat{\rho}}_{SBS,t}$. In practice, it is easier to bound $\big\| \boldsymbol{\hat{\rho}}_{t} - \boldsymbol{\hat{\rho}}_{PSBS,t}   \big\|_{1}$ and then use Lemma \ref{eqn:reversetin}, stated below, to bound $\big\| \boldsymbol{\hat{\rho}}_{t} - \boldsymbol{\hat{\rho}}_{SBS,t} \big\|_{1}$.
\begin{definition2}
\label{eqn:reversetin}
For density operators $\boldsymbol{\hat{\rho}}$ and $\boldsymbol{\hat{\sigma}}$, and constants $L\geq 0$, $\eta \in [0,1]$, $\big\|\boldsymbol{\hat{\rho}}-\eta\boldsymbol{\hat{\sigma}}\big\|_{1}\leq L$ implies $\big\|\boldsymbol{\hat{\rho}}-\boldsymbol{\hat{\sigma}}\big\|_{1}\leq 2L$ 
\end{definition2}
\begin{proof}
Using reverse triangle inequality we see that 
\begin{equation}
L\geq \|\boldsymbol{\hat{\rho}}-\eta\boldsymbol{\hat{\sigma}}\|_{1}\geq \big|\|\boldsymbol{\hat{\rho}}\|_{1}-\|\eta\boldsymbol{\hat{\sigma}}\|_{1}\big| = \|\boldsymbol{\hat{\rho}}\|_{1}-\|\eta\boldsymbol{\hat{\sigma}}\|_{1} = 1-\eta
\end{equation}
furthermore
\begin{equation}
 \|\boldsymbol{\hat{\rho}}-\boldsymbol{\hat{\sigma}}\|_{1} =   \|\boldsymbol{\hat{\rho}}-\eta\boldsymbol{\hat{\sigma}}+\eta\boldsymbol{\hat{\sigma}}-\boldsymbol{\hat{\sigma}}\|_{1}   \leq   \|\boldsymbol{\hat{\rho}}-\eta\boldsymbol{\hat{\sigma}}\|_{1}+\|\eta\boldsymbol{\hat{\sigma}}-\boldsymbol{\hat{\sigma}}\|_{1} \leq
\end{equation}
\begin{equation}
L +(1-\eta)\|\boldsymbol{\hat{\sigma}}\|_{1} = L+(1-\eta)\leq L+L= 2L
\end{equation}
\end{proof}
We now prove a preliminary inequality. 
\begin{equation}
\label{eqn:start27}
 \big\|  \boldsymbol{\hat{\rho}}_{t} - \boldsymbol{\hat{\rho}}_{UN,t}   \big\|_{1}= 
\end{equation}
\begin{equation}
\bigg\|\sum_{i,j=1}^{d_{S}}\sigma_{i,j}\Gamma(i,j,t)|i\rangle\langle j| \otimes\bigotimes_{k=1}^{N_{E}} \boldsymbol{\hat{\rho}}_{x_{i},x_{j}}^{E^{k},t}- \sum_{i=1}^{d_{S}}\sigma_{i}|i\rangle\langle i| \otimes\bigotimes_{k=1}^{N_{E}} \mathbf{\hat{P}}_{i}^{E^{k},t}\boldsymbol{\hat{\rho}}_{x_{i}}^{E^{k},t}\mathbf{\hat{P}}_{i}^{E^{k},t}\bigg\|_{1}\leq
\end{equation}
\begin{equation}
\bigg\|\sum_{i=1}^{d_{S}}\sigma_{i}|i\rangle\langle i| \otimes\bigotimes_{k=1}^{N_{E}} \boldsymbol{\hat{\rho}}_{x_{i}}^{E^{k},t}-\ \sum_{i=1}^{d_{S}}\sigma_{i}|i\rangle\langle i| \otimes\bigotimes_{k=1}^{N_{E}}\mathbf{\hat{P}}_{i}^{E^{k},t}\boldsymbol{\hat{\rho}}_{x_{i}}^{E^{k},t} \mathbf{\hat{P}}_{i}^{E^{k},t}\bigg\|_{1}+\bigg\|\sum_{i}\sum_{j;j\neq i}^{d_{S}}\sigma_{i,j}\Gamma(i,j,t)|i\rangle\langle j| \otimes\bigotimes_{k=1}^{N_{E}} \boldsymbol{\hat{\rho}}_{x_{i},x_{j}}^{E^{k},t}\bigg\|_{1} \leq
\end{equation}
\begin{equation}
\sum_{i=1}^{d_{S}}\bigg\|\sigma_{i}|i\rangle\langle i| \otimes\bigotimes_{k=1}^{N_{E}} \boldsymbol{\hat{\rho}}_{x_{i}}^{E^{k},t}- \sigma_{i}|i\rangle\langle i| \otimes\bigotimes_{k=1}^{N_{E}} \mathbf{\hat{P}}_{i}^{E^{k},t}\boldsymbol{\hat{\rho}}_{x_{i}}^{E^{k},t} \mathbf{\hat{P}}_{i}^{E^{k},t}\bigg\|_{1}+\bigg\|\sum_{i}\sum_{j; j\neq i}^{d_{S}}\sigma_{i,j}\Gamma(i,j,t)|i\rangle\langle j| \otimes\bigotimes_{k=1}^{N_{E}} \boldsymbol{\hat{\rho}}_{x_{i},x_{j}}^{E^{k},t}\bigg\|_{1}\leq
\end{equation}
\begin{equation}
\sum_{i=1}^{d_{S}}\sigma_{i}\bigg\| |i\rangle\langle i| \otimes\bigg(\bigotimes_{k=1}^{N_{E}}\boldsymbol{\hat{\rho}}_{x_{i}}^{E^{k},t}- \bigotimes_{k=1}^{N_{E}}\mathbf{\hat{P}}_{i}^{E^{k},t}\boldsymbol{\hat{\rho}}_{x_{i}}^{E^{k},t} \mathbf{\hat{P}}_{i}^{E^{k},t}\bigg)\bigg\|_{1}+\sum_{i}\sum_{j;j\neq i}^{d_{S}}\big|\sigma_{i,j}\Gamma(i,j,t)\big|\bigg\||i\rangle\langle j| \otimes\bigotimes_{k=1}^{N_{E}} \boldsymbol{\hat{\rho}}_{x_{i},x_{j}}^{E^{k},t}\bigg\|_{1}=
\end{equation}
\begin{equation}
\sum_{i=1}^{d_{S}}\sigma_{i}\bigg\|\bigotimes_{k=1}^{N_{E}}\boldsymbol{\hat{\rho}}_{x_{i}}^{E^{k},t}-\bigotimes_{k=1}^{N_{E}}\mathbf{\hat{P}}_{i}^{E^{k},t}\boldsymbol{\hat{\rho}}_{x_{i}}^{E^{k},t} \mathbf{\hat{P}}_{i}^{E^{k},t}\bigg\|_{1}+\sum_{i}\sum_{j; j \neq i}^{d_{S}}\big|\sigma_{i,j}\Gamma(i,j,t)\big| \leq
\end{equation}
\begin{equation}
\sum_{k=1}^{N_{E}}\sum_{i=1}^{d_{S}}\sigma_{i}\bigg\|\boldsymbol{\hat{\rho}}_{x_{i}}^{E^{k}_{t}}-\mathbf{\hat{P}}_{i}^{E^{k},t}\boldsymbol{\hat{\rho}}_{x_{i}}^{E^{k},t} \mathbf{\hat{P}}_{i}^{E^{k},t}\bigg\|_{1}+\sum_{i}\sum_{j; j \neq i}^{d_{S}}\big|\sigma_{i,j}\Gamma(i,j,t)\big|   
\end{equation}
where in the last step we have used Lemma \ref{eqn:telescoping}. Using Lemma \ref{eqn:reversetin} we conclude that 
\begin{equation}
\label{eqn:result27}
 \frac{1}{2}\min_{PVM}\big\|  \boldsymbol{\hat{\rho}}_{t} - \boldsymbol{\hat{\rho}}_{SBS,t} \big\|_{1}\leq
\min_{PVM}\bigg(\sum_{k=1}^{N_{E}}\sum_{i=1}^{d_{S}}\sigma_{i}\bigg\|  \boldsymbol{\hat{\rho}}_{x_{i}}^{E^{k},t}-\mathbf{\hat{P}}_{i}^{E^{k},t}\boldsymbol{\hat{\rho}}_{x_{i}}^{E^{k},t} \mathbf{\hat{P}}_{i}^{E^{k},t}\bigg\|_{1}\bigg)+\Gamma(t)
\end{equation}

 $\Gamma(t) :=  \sum_{i}\sum_{j;j\neq i}|\sigma_{i,j}| \prod_{k=N_{E}+1}^{N}| \gamma_{i,j}^{k}(t)|$, $\gamma_{i,j}^{k}(t)=Tr\big[ \boldsymbol{\hat{\rho}}_{x_{i},x_{j}}^{E^{k},t}  \big]$, and $\sigma_{i,j} := \langle i|\boldsymbol{\hat{\rho}}_{S_{0}}|j\rangle$.
In (\ref{eqn:result27}), $\Gamma(t)$ is the decoherence term which is independent of the choice of the PVM minimized over. The decoherence factor is simple to study provided that we are able to compute the traces defining the terms $\gamma^{k}_{i,j}(t)$. The first term in (\ref{eqn:result27}) involves a minimization over all PVM for each value of $t$. Rather than attempting to solve the minimization problem exactly, we will use Theorem \ref{eqn:theoremgram} to bound (\ref{eqn:result27}).\\

To do the latter we assume that the initial states $\boldsymbol{\hat{\rho}}^{E^{k}}$ are pure; the general case will be considered in Section \ref{sec:mixedmixed}. The purity of $\boldsymbol{\hat{\rho}}^{E^{k}}$ implies that the operators $\boldsymbol{\hat{\rho}}_{i}^{E^{k},t}$ ( using the notation defined in (\ref{eqn:purepreserve})) are pure as well for all $i$ since the evolution (\ref{eqn:purepreserve}) is unitary. We will henceforth write $\boldsymbol{\hat{\rho}}^{E^{k}}$ as a projector. 
\begin{equation}
\label{eqn:adef35}
\big|\psi^{k}_{i,t}\big\rangle\big\langle \psi^{k}_{i,t}\big|= \boldsymbol{\hat{\rho}}_{i}^{E^{k},t}
\end{equation}
We now use Theorem \ref{eqn:theoremgram} to estimate the first term (\ref{eqn:result27}), obtaining the following theorem. 
\begin{definition4}
\label{th:thegap}
Using the notation introduced earlier in this section so far,
\begin{equation}
\label{eqn:boundwithbothterms(53)}
  \frac{1}{2}\min_{POVM}\big\|  \boldsymbol{\hat{\rho}}_{t} - \boldsymbol{\hat{\rho}}_{SBS,t}   \big\|_{1}\leq \frac{1}{2}\sum_{k=1}^{N_{E}}\sum_{i=2}^{d_{S}}\sigma_{i}\sum_{s=1}^{i-1}\bigg|\frac{M_{s,i}^{k}}{D^{k}_{s-1,t}D^{k}_{s,t}}\bigg|+\Gamma(t)
\end{equation}

where
\begin{equation}
M_{s,i}^{k}:=\mathbf{det}\begin{pmatrix}
\big\langle \psi^{k}_{1,t}\big|\psi^{k}_{1,t}\big\rangle  & \big\langle \psi^{k}_{1,t}\big|\psi^{k}_{2,t}\big\rangle & \dots & \big\langle\psi^{k}_{1,t}\big|\psi^{k}_{s,t}\big\rangle\\

\big\langle \psi^{k}_{2,t}\big|\psi^{k}_{1,t}\big\rangle & \big\langle \psi^{k}_{2,t}\big|\psi^{k}_{2,t}\big\rangle & \dots & \big\langle \psi^{k}_{2,t}\big|\psi^{k}_{s,t}\big\rangle\\

\vdots & \vdots & \ddots & \vdots\\

\big\langle \psi^{k}_{s-1,t}\big|\psi^{k}_{1,t}\big\rangle & \big\langle \psi^{k}_{s-1,t}\big|\psi^{k}_{2,t}\big\rangle& \dots & \big\langle \psi^{k}_{s-1,t}\big|\psi^{k}_{s,t}\big\rangle\\
\big\langle \psi^{k}_{i,t}\big|\psi^{k}_{1,t}\big\rangle & \big\langle \psi^{k}_{i,t}\big|\psi^{k}_{2,t}\big\rangle & \dots &\big\langle \psi^{k}_{i,t}\big|\psi^{k}_{s,t}\big\rangle\\
\end{pmatrix} 
\end{equation}
\begin{equation}
\label{eqn:theequationstuff}
\hspace{0.3cm}D^{k}_{s,t} : = \mathbf{det}\begin{pmatrix}
\big\langle \psi^{k}_{1,t}\big|\psi^{k}_{1,t}\big\rangle  & \big\langle \psi^{k}_{1,t}\big|\psi^{k}_{2,t}\big\rangle & \dots & \big\langle\psi^{k}_{1,t}\big|\psi^{k}_{s,t}\big\rangle\\

\big\langle \psi^{k}_{2,t}\big|\psi^{k}_{1,t}\big\rangle & \big\langle \psi^{k}_{2,t}\big|\psi^{k}_{2,t}\big\rangle & \dots & \big\langle \psi^{k}_{2,t}\big|\psi^{k}_{s,t}\big\rangle\\

\vdots & \vdots & \ddots & \vdots\\

\big\langle \psi^{k}_{j,t}\big|\psi^{k}_{1,t}\big\rangle & \big\langle \psi^{k}_{j,t}\big|\psi^{k}_{2,t}\big\rangle & \dots &\big\langle \psi^{k}_{s,t}\big|\psi^{k}_{s,t}\big\rangle\\
\end{pmatrix}
\end{equation}
\end{definition4}

Theorem \ref{th:thegap} is an alternative to corollary 1 of \cite{JKthree} which remains unproven.

\section{Mixed Environmental States}
\label{sec:mixedmixed}
\;\;\; The theory we have developed so far considers only the case where $\boldsymbol{\hat{\rho}}_{x_{i}}^{E_{t}^{k}}$ are pure states for all $i$ and $k$. In this section, we will provide the analog of Theorem \ref{eqn:theoremgram} in the case where the environmental degrees of freedom are finite mixtures of pure states. Using a simpler indexing scheme, consider a mixed state of the form $\sum_{i=1}^{N}p_{i}\boldsymbol{\hat{\rho}}_{i}$, where $\sum_{i=1}^N p_i = 1$ and the $\boldsymbol{\hat{\rho}}_i$ are all countably-mixed states; i.e. $\boldsymbol{\hat{\rho}}_{i}=\sum_{k=1}^{M_{i}}\eta_{ik}\boldsymbol{\hat{\rho}}_{ik}$ where all of the $\boldsymbol{\hat{\rho}}_{ik}$ are pure states and $\sum_{k=1}^{M_{i}}\eta_{ik} =1$.
Let us now consider the QSD problem (Appendix \ref{app:QSD}) of finding 
\begin{equation}
\label{eqn:POVMmin1}
\min_{POVM}\sum_{i=1}^{N}p_{i}Tr\Big\{ \boldsymbol{\hat{\rho}}_{i}-\mathbf{\hat{M}}_{i}\boldsymbol{\hat{\rho}}_{i}\mathbf{\hat{M}}_{i}^{\dagger} \Big\}.
\end{equation}
We obtain an upper bound on (\ref{eqn:POVMmin1}) if we minimize over all PVM instead of all POVM.
\begin{equation}
\label{eqn:POVMmin2}
   \min_{POVM}\sum_{i=1}^{N}p_{i}Tr\Big\{ \boldsymbol{\hat{\rho}}_{i}-\mathbf{\hat{M}}_{i}\boldsymbol{\hat{\rho}}_{i}\mathbf{\hat{M}}_{i}^{\dagger}  \Big\} \leq \min_{PVM}\sum_{i=1}^{N}p_{i}Tr\Big\{ \boldsymbol{\hat{\rho}}_{i}-\mathbf{\hat{M}}_{i}\boldsymbol{\hat{\rho}}_{i}\mathbf{\hat{M}}_{i}^{\dagger} \Big\}\leq
\end{equation}   
\begin{equation} 
\label{eqn:dragon5}
\min_{PVM}\sum_{i=1}^{N}p_{i}\Big\|\boldsymbol{\hat{\rho}}_{i}-\mathbf{\hat{M}}_{i}\boldsymbol{\hat{\rho}}_{i}\mathbf{\hat{M}}_{i}^{\dagger} \Big\|_{1}
\end{equation}

The left-hand side of (\ref{eqn:POVMmin2}) can be bounded from below using Theorem \ref{eqn:monta} . 
Namely,
\begin{equation}
\label{eqn:POVMmin3}
\frac{1}{2}\sum_{i=1}^{N}\sum_{j; j\neq i}^{N}p_{i}p_{j}F(\boldsymbol{\hat{\rho}}_{i},\boldsymbol{\hat{\rho}}_{j})\leq\min_{PVM}\sum_{i=1}^{N}p_{i}\Big\|\boldsymbol{\hat{\rho}}_{i}-\mathbf{\hat{P}}_{i}\boldsymbol{\hat{\rho}}_{i}\mathbf{\hat{P}}_{i}\Big\|_{1}
\end{equation}
Expanding the $\boldsymbol{\hat{\rho}}_{i}$ we see that
\begin{equation}
\label{eqn:concavityjoint}
\sqrt{F(\boldsymbol{\hat{\rho}}_{i},\boldsymbol{\hat{\rho}}_{j})} = \sqrt{F\Big(\sum_{k=1}^{M_{i}}\eta_{ik}\boldsymbol{\hat{\rho}}_{ik},\sum_{k=1}^{M_{j}}\eta_{jk}\boldsymbol{\hat{\rho}}_{jk}\Big)} \geq \sum_{k=1}^{\min\{M_{i}, M_{j}\}}\sqrt{\eta_{ik}\eta_{jk}}\sqrt{F\Big(\boldsymbol{\hat{\rho}}_{ik}, \boldsymbol{\hat{\rho}}_{jk}\Big)}
\end{equation}
where we have used Theorem 9.7 of \cite{Nielsen} in (\ref{eqn:concavityjoint}). (\ref{eqn:POVMmin3}) now implies that 
\begin{equation}
\label{eqn:POVMmin4}
\frac{1}{2}\sum_{i=1}^{N}\sum_{j;j\neq i}^{N}p_{i}p_{j}\Bigg(\sum_{k=1}^{\min\{M_{i}, M_{j}\}}\sqrt{\eta_{ik}\eta_{jk}}\sqrt{F\Big(\boldsymbol{\hat{\rho}}_{ik}, \boldsymbol{\hat{\rho}}_{jk}\Big)}\Bigg)^{2}\leq\min_{PVM}\sum_{i=1}^{N}p_{i}\Big\| \boldsymbol{\hat{\rho}}_{i}-\mathbf{\hat{P}}_{i}\boldsymbol{\hat{\rho}}_{i}\mathbf{\hat{P}}_{i}\Big\|_{1}
\end{equation}
This inequality shows that a necessary condition for fully solving the optimization problem (\ref{eqn:POVMmin4}) is that $F\big(\boldsymbol{\hat{\rho}}_{ik}, \boldsymbol{\hat{\rho}}_{jk}\big) = 0 $ for all $i,j,k$ where $i\neq j$. Otherwise, the minimum is not equal to zero. In the case where the $\boldsymbol{\hat{\rho}}_{i}$ are not mixed states the corresponding condition is $F\big(\boldsymbol{\hat{\rho}}_{i},\boldsymbol{\hat{\rho}}_{j}\big) = 0 $ for $i\neq j$, which is what we expect for the states studied in this paper. For the case where the $\boldsymbol{\hat{\rho}}_{i}$ are finite mixtures of pure states, one will be required to analyze the fidelities between elements of any two different mixtures, say $\boldsymbol{\hat{\rho}}_{i}$ and $\boldsymbol{\hat{\rho}}_{j}$, in order to determine the discriminability of the mixture $\sum_{i=1}^{N}\boldsymbol{\hat{\rho}}_{i}$. We will analyze these types of mixtures in the following.

We will be estimating the right-hand side of (\ref{eqn:POVMmin4}). Our approach shall be an adaptation of the methods employed in the proof of Theorem \ref{eqn:theoremgram} and Lemma \ref{eqn:lemmaboundwithgramschmidt}. Under the assumption that $\mathbf{\hat{P}}_{i}$ are projectors we will prove a bound that will be useful for the cases where $F\big(\boldsymbol{\hat{\rho}}_{ik},\boldsymbol{\hat{\rho}}_{jk}\big) = 0 $ for all $k$ when $i\neq j$ and $F\big(\boldsymbol{\hat{\rho}}_{ik},\boldsymbol{\hat{\rho}}_{il}\big) = 0 $ for all $i$ and $l\neq k$ are satisfied only approximately.\\ 

We begin by noting that 
\begin{equation}
\label{eqn:dragon6}
\min_{PVM}\sum_{i=1}^{N}p_{i}\Big\| \boldsymbol{\hat{\rho}}_{i}-\mathbf{\hat{P}}_{i}\boldsymbol{\hat{\rho}}_{i}\mathbf{\hat{P}}_{i}\Big\|_{1} \leq\min_{PVM}\sum_{i=1}^{N}\sum_{k=1}^{M_{i}}p_{i}\eta_{ik}\Big\| \boldsymbol{\hat{\rho}}_{ik}-\mathbf{\hat{P}}_{i}\boldsymbol{\hat{\rho}}_{ik}\mathbf{\hat{P}}_{i}\Big\|_{1} 
\end{equation}
This looks very similar to the PVM QSD problem for pure states considered in the previous two sections. Note that $p_{i}\eta_{ik}$, $1\leq i\leq N$ and $1\leq k\leq M_{i}$, is a probability distribution, but now each element of the PVM $\big\{\mathbf{\hat{P}}_{i}\big\}_{i}$ multiplies all elements $\boldsymbol{\hat{\rho}}_{ik}$ of the $i$th mixture. Following the methods from the previous section, one might suggest implementing the Gram-Schmidt procedure once more in order to obtain an orthonormal set of vectors $\big|\phi_{i}\big\rangle$. However, in this case, the operators $\boldsymbol{\hat{\rho}}_{i}$ are mixed and therefore do not correspond to vectors in the corresponding Hilbert spaces; being able to view the operators  $\boldsymbol{\hat{\rho}}_{i}$ as an ensemble of pure states was one of the key assumptions that led to Theorem \ref{eqn:theoremgram}.\\

We now assume that the $\mathbf{\hat{P}}_{i}$ have the following structure.
\begin{equation}
   \mathbf{\hat{P}}_{i} = \sum_{k=1}^{M_{i}}\mathbf{\hat{P}}_{ik}
\end{equation}
In order to guarantee that the $\sum_{k=1}^{M_{i}}\mathbf{\hat{P}}_{ik}$ is a projector, we will assume that the ranges $\mathbf{\hat{P}}_{ik}$ are orthogonal to each other, $k$ varying from $1$ to $M_{i}$.

Since all of the $\boldsymbol{\hat{\rho}}_{ik}$ are pure states, we may apply the Gram-Schmidt process in order to construct a PVM $\big\{ \mathbf{\hat{P}}_{ik}\big\}_{ik}$. 
Together with $\mathbb{I}-\sum_{i}\sum_{k}\mathbf{\hat{P}}_{ik}$, these operators form a PVM that resolves the identity. 
There are $\sum_{i=1}^{N}M_{i}$ states $\boldsymbol{\hat{\rho}}_{ik}$, let us now arrange them in a sequence
\begin{equation}\
\label{eqn:seqseq}
\vec{\mathscr{V}}:=  \begin{pmatrix}
\boldsymbol{\hat{\rho}}_{11}  & \dots & \boldsymbol{\hat{\rho}}_{1M_{1}}&
\boldsymbol{\hat{\rho}}_{21}  & \dots & \boldsymbol{\hat{\rho}}_{2M_{2}}
 &  \dots &\boldsymbol{\hat{\rho}}_{N1}    & \dots\dots & \boldsymbol{\hat{\rho}}_{NM_{N}}\\
\end{pmatrix}.
\end{equation}
Let us name the $s$th component of this sequence $\mathscr{V}_{s}:= \big|\xi_{s}\big\rangle\big\langle \xi_{s}\big|$. Since the cardinality of the sequence (\ref{eqn:seqseq}) is the same as the cardinality of the set $\{\boldsymbol{\hat{\rho}}_{ik}\}_{ik}$, $(1\leq i\leq N, 1\leq k \leq M_{i})$, there exists a bijection $\mathscr{M}$ mapping every pair $ik$ to a unique s  and viceversa.

Assuming that the $\big|\xi_{s}\big\rangle$ form a linearly independent set we apply the Gram-Schmidt process to obtain the family of orthonormal states
\begin{equation}
\label{eqn:dragon15}
\big|\phi_{1}\big\rangle:=\big|\xi_{1}\big\rangle
\end{equation}
\begin{equation}
\label{eqn:dragon1515}
\big|\phi_{s}\big\rangle := \frac{1}{\alpha_{s}}\bigg\{\big|\xi_{s}\big\rangle - \sum_{k=1}^{s-1}\big\langle \phi_{k}\big|\xi_{s}\big\rangle\big|\phi_{k}\big\rangle \bigg\}, \:\: s\in\{1,2,...,\sum_{i=1}^{N}M_{i}\}
\end{equation}
were as before $\alpha_{i}: = \big\|\big|\xi_{i}\big\rangle - \sum_{k=1}^{i-1}\big\langle \phi_{k}\big|\xi_{i}\big\rangle\big|\phi_{k}\big\rangle\big\| = \sqrt{1-\sum_{k=1}^{i-1}|\big\langle\phi_{k}\big|\xi_{i}\big\rangle|^{2}}$ for $i>1$ and $\alpha_{1} = 1$ are the respective normalization constants. We have thus constructed a PVM $\Big\{\big|\phi_{s}\big\rangle\big\langle \phi_{s}\big|\Big\}_{s}\bigcup \Big\{\mathbb{I}-\sum_{s}\big|\phi_{s}\big\rangle\big\langle \phi_{s}\big| \Big\}$. Defining $\omega_{s}:=p_{i}\eta_{ik}$, where $s = \mathscr{M}\big(ik\big)$, and letting $\mathscr{P}(ik):=i$, and $\mathscr{I}:=\mathscr{P}\circ\mathscr{M}^{-1}$, we may now rewrite and bound the expression
\begin{equation}
\sum_{i=1}^{N}\sum_{k=1}^{M_{i}}p_{i}\eta_{ik}\Big\| \boldsymbol{\hat{\rho}}_{ik}-\mathbf{P}_{i}\boldsymbol{\hat{\rho}}_{ik}\mathbf{P}_{i}\Big\|_{1}
\end{equation}
as follows. 
\smaller
\begin{equation}
\sum_{s}\omega_{s}\bigg\| \big|\xi_{s}\big\rangle\big\langle\xi_{s}\big|-\Big(\sum_{\substack{l\;\ni \\\mathscr{I}(l) = \mathscr{I}(s)}}\big|\phi_{l}\big\rangle\big\langle\phi_{l}\big|\Big)\big|\xi_{s}\big\rangle\big\langle\xi_{s}\big|\Big(\sum_{\substack{l\; \ni \\\mathscr{I}(l) = \mathscr{I}(s)}}\big|\phi_{l}\big\rangle\big\langle \phi_{l}\big|\Big)\bigg\|_{1} = 
\end{equation}
\begin{equation}
\sum_{s}\omega_{s}\Bigg\| \big|\xi_{s}\big\rangle\big\langle \xi_{s}\big|-\big|\phi_{s}\big\rangle\big\langle \phi_{s}\big|\xi_{s}\big\rangle\big\langle\xi_{s}\big|\phi_{s}\big\rangle\big\langle \phi_{s}\big|-\Big(\sum_{\substack{l;l\neq s\;\ni \\\mathscr{I}(l) = \mathscr{I}(s)}}\big|\phi_{l}\big\rangle\big\langle\phi_{l}\big|\Big)\big|\xi_{s}\big\rangle\big\langle\xi_{s}\big|\Big(\sum_{\substack{k;k\neq s\; \ni \\\mathscr{I}(k) = \mathscr{I}(s)}}\big|\phi_{k}\big\rangle\big\langle \phi_{k}\big|\Big)\Bigg\|_{1} \leq
\end{equation}

\begin{equation}
\label{eqn:dragon10}
\sum_{s}\omega_{s}\Big\| \big|\xi_{s}\big\rangle\big\langle \xi_{s}\big|-\big|\phi_{s}\big\rangle\big\langle \phi_{s}\big|\xi_{s}\big\rangle\big\langle\xi_{s}\big|\phi_{s}\big\rangle\big\langle \phi_{s}\big|\Big\|_{1}+\sum_{s}\omega_{s}\Bigg\|\Big(\sum_{\substack{l;l\neq s\; \ni \\\mathscr{I}(l) = \mathscr{I}(s)}}\big|\phi_{l}\big\rangle\big\langle\phi_{l}\big|\Big)\big|\xi_{s}\big\rangle\big\langle\xi_{s}\big|\Big(\sum_{\substack{k;k\neq s\; \ni \\\mathscr{I}(k) = \mathscr{I}(s)}}\big|\phi_{k}\big\rangle\big\langle \phi_{k}\big|\Big)\Bigg\|_{1} \leq
\end{equation}
\begin{equation}
\label{eqn:dragon11}
2\sum_{s}\omega_{s}\sum_{k=1}^{s-1}|\big\langle \phi_{k}\big|\xi_{s}\big\rangle|+\sum_{s}\omega_{s}\sum_{\substack{l;l\neq s\; \ni \\\mathscr{I}(l) = \mathscr{I}(s)}}\Bigg(\sum_{\substack{k;k\neq s\; \ni \\\mathscr{I}(k) = \mathscr{I}(s)}}\Big\|\big|\phi_{l}\big\rangle\big\langle\phi_{l}\big|\xi_{s}\big\rangle\big\langle\xi_{s}\big|\phi_{k}\big\rangle\big\langle \phi_{k}\big|\Big\|_{1}\Bigg)= 
\end{equation}
\begin{equation}
\label{eqn:dragon12}
2\sum_{s}\omega_{s}\sum_{k=1}^{s-1}\big|\big\langle \phi_{k}\big|\xi_{s}\big\rangle\big|+\sum_{s}\omega_{s}\sum_{\substack{l;l\neq s\; \ni \\\mathscr{I}(l) = \mathscr{I}(s)}}\Bigg(\sum_{\substack{k;k\neq s\; \ni \\\mathscr{I}(k) = \mathscr{I}(s)}}\big|\big\langle\phi_{l}\big|\xi_{s}\big\rangle\big\langle\xi_{s}\big|\phi_{k}\big\rangle\big|\Bigg)
\end{equation}
\normalsize
where we have used Lemma \ref{eqn:lemmaboundwithgramschmidt} in going from (\ref{eqn:dragon10}) to (\ref{eqn:dragon11}). Using Lemma \ref{eqn:lemmagramschmidtmatrix} we may explicitly write the terms $|\big\langle \phi_{l}\big|\xi_{s}\big\rangle| $ as Gram determinants and use these to estimate the efficacy of the PVM built from (\ref{eqn:dragon15}) and (\ref{eqn:dragon1515}). $(\ref{eqn:dragon12}) $ may be further bounded as follows. 
\begin{equation}
\label{end}
(\ref{eqn:dragon12}) \leq 3\sum_{s}\omega_{s}\sum_{\substack{l;l\neq s\; \ni \\\mathscr{I}(l) = \mathscr{I}(s)}}|\big\langle \phi_{l}\big|\xi_{s}\big\rangle|
\end{equation}
As already mentioned, this may be better estimated using Lemma \ref{eqn:lemmagramschmidtmatrix}. We state the result (\ref{eqn:dragon12}) as a theorem.\\

\begin{definition4}
\label{eqn:mainrest}
Consider a mixed state of the form $\sum_{i=1}^{N}p_{i}\boldsymbol{\hat{\rho}}_{i}$,\:$\sum_{i=1}^{N}p_{i}=1$, where $\boldsymbol{\hat{\rho}}_i$ are countable mixtures of pure states,\, i.e. $\boldsymbol{\hat{\rho}}_{i}=\sum_{k=1}^{M_{i}}\eta_{ik}\big|\psi_{ik}\big\rangle\big\langle \psi_{ik}\big|$ where $\sum_{k=1}^{M_{i}}\eta_{ik} =1$. Furthermore, assume that the set $\big\{\big|\psi_{ik}\big\rangle\big\}_{ik}$ is linearly independent. Then
\begin{equation}
\min_{PVM}\sum_{i=1}^{N}p_{i}\Big\| \boldsymbol{\hat{\rho}}_{i}-\mathbf{\hat{P}}_{i}\boldsymbol{\hat{\rho}}_{i}\mathbf{\hat{P}}_{i}\Big\|_{1} \leq   
\end{equation}
\begin{equation}
2\sum_{s}\omega_{s}\sum_{k=1}^{s-1}|\big\langle \phi_{k}\big|\xi_{s}\big\rangle|+\sum_{s}\omega_{s}\sum_{\substack{l;l\neq s\; \ni \\\mathscr{I}(l) = \mathscr{I}(s)}}\Bigg(\sum_{\substack{k;k\neq s\; \ni \\\mathscr{I}(k) = \mathscr{I}(s)}}|\big\langle\phi_{l}\big|\xi_{s}\big\rangle\big\langle\xi_{s}\big|\phi_{k}\big\rangle|\Bigg)
\end{equation}
where 
\begin{equation}
\big|\xi_{s}\big\rangle:= \big|\psi_{\mathscr{M}(ik)}\big\rangle
\end{equation}
and $\mathscr{M}$ was defined in the discussion following (\ref{eqn:seqseq})
\begin{equation}
\big|\phi_{1}\big\rangle:=\big|\xi_{1}\big\rangle
\end{equation}
\begin{equation}
\big|\phi_{s}\big\rangle := \frac{1}{\alpha_{s}}\bigg\{\big|\xi_{s}\big\rangle - \sum_{k=1}^{s-1}\big\langle \phi_{k}\big|\xi_{s}\big\rangle\big|\phi_{k}\big\rangle \bigg\}, \:\: s\in\{1,2,...,S:=\sum_{i=1}^{N}M_{i}\}
\end{equation}
and 
\begin{equation}
\alpha_{i}: = \big\|\big|\xi_{i}\big\rangle - \sum_{k=1}^{i-1}\big\langle \phi_{k}\big|\xi_{i}\big\rangle\big|\phi_{k}\big\rangle\big\| = \sqrt{1-\sum_{k=1}^{i-1}|\big\langle\phi_{k}\big|\xi_{i}\big\rangle|^{2}}
\end{equation}
for $i>1$ and $\alpha_{1} = 1$ are the respective normalizing constants.
\end{definition4}
\begin{proof}
The proof follows from the preceding discussion.   
\end{proof}

\section{How general may the $\mathbf{\hat{B}}_{k}$ be?}
\label{eqn:howgen}
\;\;\; We conclude this subsection with the following corollary. The question it sheds light on is: "How general may the $\mathbf{\hat{B}}_{k}$ be and still induce dynamics (\ref{eqn:decmodechan}) leading to SBS states?". Before proceeding we state the following definition. 
\begin{definition}[Rajchman Measure:]
A finite Borel probability measure $\mu$ on $\mathbf{R}$ is called a Rajchman measure if it satisfies 
\begin{equation}
\lim_{|t|\rightarrow\infty}\hat{\mu}(t) = 0 
\end{equation}
where
$\hat{\mu}(t) := \int_{\mathbb{R}} e^{2i\pi tx}d\mu(x)$ \;\;, $t\in \mathbb{R}$.
\end{definition}

\begin{definition4}[The Rajchman Subspace is a Closed Subspace]
\label{eqn:firstresult}
Let $\mathbf{\hat{A}}$ (let $\mathbf{\hat{A}}$ be in general an unbounded operator on a $\mathscr{H}_{S}$) be a self-adjoint operator acting on some arbitrary Hilbert space $\mathscr{H}$, then the set of vectors in $\mathscr{H}$ for which the spectral measure is a Rajchman measure, i.e.
\begin{equation}
\label{eqn:link2}
 \mathscr{H}_{rc}:=\Big\{\big|\psi\big\rangle\;| \;\lim_{|t|\rightarrow \infty} \big\langle \psi\big|e^{-it\mathbf{\hat{A}}}\big|\psi\big\rangle = 0\Big\},
\end{equation}
is a closed subspace that is invariant under the unitary group $e^{-is\mathbf{\hat{A}}}$\cite{gerald}.
\end{definition4}
\begin{Co}[Sufficient conditions for the convergence to SBS for a broad family of multipartite states]
\label{eqn:zelda1000}
Consider the setup contained in equations (\ref{eqn:snake}) through (\ref{eqn:decmodechan}). If for all $k$, $\mathbf{\hat{B}}_{k}$ has a non-empty Rajchman subspace $\mathscr{H}_{E^{k},rc}$ (\cite{gerald}), and $\boldsymbol{\hat{\rho}}^{E^{k}}$ is a finite mixture of pure states in $\mathcal{S}(\mathscr{H}_{E^{k},rc})$, then $\Lambda_{t}\big(\boldsymbol{\hat{\rho}}\big)$ is SBS-asymptotic.
\end{Co}
\begin{proof}
Using Theorems \ref{th:thegap} and \ref{eqn:mainrest}, $\sum_{k=1}^{N_{E}}\sum_{i=1}^{d_{S}}\sigma_{i}\bigg\|\boldsymbol{\hat{\rho}}_{x_{i}}^{E^{k},t}-\mathbf{\hat{P}}_{i}^{E^{k},t}\boldsymbol{\hat{\rho}}_{x_{i}}^{E^{k},t}\mathbf{\hat{P}}_{i}^{E^{k},t}\bigg\|_{1} $ may be estimated using inner products of distinct $\big|\psi_{ik}\rangle$ ( see (\ref{eqn:morocco4}) and (\ref{eqn:theend7})). Furthermore, for all $i\neq j$, $\Gamma(i,j,t)$ (\ref{eqn:thegamm}) is a product of inner products. The inner products featured in both diagonal terms
\begin{equation}
\sum_{k=1}^{N_{E}}\sum_{i=1}^{d_{S}}\sigma_{i}\bigg\|\boldsymbol{\hat{\rho}}_{x_{i}}^{E^{k},t}-\mathbf{\hat{P}}_{i}^{E^{k},t}\boldsymbol{\hat{\rho}}_{x_{i}}^{E^{k},t} \mathbf{\hat{P}}_{i}^{E^{k},t}\bigg\|_{1}\end{equation}
and the off-diagonal terms $\sum_{i}\sum_{j; j \neq i}^{d_{S}}\big|\sigma_{i,j}\Gamma(i,j,t)\big|$ have the structure $\big\langle \psi\big|e^{-it\alpha\mathbf{\hat{B}}_{k}}\big|\phi\big\rangle$ with $\big|\psi\big\rangle\in \mathscr{H}_{E^{k},rc}$ which implies the claim (using the fact that the Rajchman subspace is invariant under $e^{-is\mathbf{\hat{B}}_{k}}$). 
\end{proof}
\newpage

\section{Conclusions}
In this paper, we have analyzed a generic finite-dimensional open system with von Neumann measurement-type interactions between the central system and the environment. We studied the SBS-asymptotics of such multipartite systems. Working with pure
environment states, we have used the idea of ``self-referencing'', i.e. building from the actual state its SBS version with the help of Gram-Schmidt procedure.  The approach of a state to its SBS limit can be monitored by bounds on their trace distance. Furthermore, we also derive a version of our condition for mixed environmental states which we plan to strengthen in the future.  Another future direction is the application of our scheme to the continuous variable case, i.e. to the case where the observable that is tested for being objectively encoded has a continuous spectrum; e.g. position.

\section{Acknowledgements}
J.K.K acknowledges the support of the Polish National Science Center (NCN) under the Grant No. 2019/35/B/ST2/01896. A.A. and J.W. were partially supported by NSF grant DMS 1911358.
J. W. was also partially supported by the Simons Foundation Fellowship 823539.

\newpage
\appendix 
\addcontentsline{toc}{section}{Appendices}
\section{Quantum State Discrimination}
\label{app:QSD}
We define the Quantum State Discrimination optimization problem (QSD)  \cite{hellstrom} \cite{Montanaro} \cite{qiu} \cite{bae} \cite{barnett}.  Let $\mathscr{H}$ be an arbitrary Hilbert space and let $\mathcal{S}(\mathcal{\mathscr{H}}\big)$ be the space of density operators acting in $\mathscr{H}$. Given a mixture of density operators, 
\begin{equation}
\label{eqn:countablemixture}
\boldsymbol{\hat{\rho}} = \sum_{i=1}^{N}p_{i} \boldsymbol{\hat{\rho}}_{i} 
\end{equation}
where $\sum_{i = 1}^{N}p_{i} = 1$, the theory of QSD aims to find a POVM $\{\mathbf{\hat{M}}_{l}^{\dagger}\mathbf{\hat{M}}_{l}\}_{l=1}^{K}\subset \mathcal{B}(\mathscr{H})$ ( $K\geq N$) which resolves the identity operator of $\mathcal{B}(\mathscr{H})$, and minimizes the object below which we will be referring to as a \emph{probability error}. 
\begin{equation}
\label{eqn:minerror2}
p_{E}\big\{\{p_{i},\boldsymbol{\hat{\rho}}_{i}\}_{i=1}^{N}, \{\mathbf{\hat{M}}_{l}\big\}_{l=1}^{K} \big\}:=1-\sum_{i=1}^{N}p_{i}Tr\big\{\mathbf{\hat{M}}_{i}\boldsymbol{\hat{\rho}}_{i}\mathbf{\hat{M}}^{\dagger}_{i} \big\} 
\end{equation}


\end{document}